\documentclass{amsart}

\usepackage{amssymb}
\usepackage{xcolor}
\usepackage{epsfig}
\usepackage{calrsfs}
\usepackage[hidelinks]{hyperref}
\theoremstyle{plain}
\begingroup
\newtheorem*{theorem*}{Theorem}

\newtheorem*{theoremB'}{Theorem B'}
\newtheorem*{theoremB''}{Theorem B''}

\newtheorem{theorem}{Theorem}[section]

\newtheorem{lemma}[theorem]{Lemma}
\newtheorem{proposition}[theorem]{Proposition}

\newtheorem{claim}[theorem]{Claim}
\newtheorem{definition}[theorem]{Definition}
\newtheorem{rmk}[theorem]{Remark}

\newtheorem{question}[theorem]{Question}

\newtheorem{open}[theorem]{Open Problem}
\endgroup

\theoremstyle{remark}
\begingroup
\endgroup

\numberwithin{equation}{section}

\renewcommand{\textbf}[1]{\begingroup\bfseries\mathversion{bold}#1\endgroup}
\newcommand{\op}[1]{{\rm{#1}}}

\newcommand{\Z}{\mathbb Z}
\newcommand{\N}{\mathbb N}
\newcommand{\R}{\mathbb R}

\newcommand{\psum}[1]{\sum_{x\in L}{\vphantom{\sum}}'}
\newcounter{margcount} 
\title[Optimal lattice and non-c.m. functions]{Optimal and non-optimal lattices for non-completely monotone interaction potentials}

\author{Laurent B\'etermin}
\address{QMATH, Department of Mathematical Sciences, University of Copenhagen, Universitetsparken 5, DK-2100 Copenhagen \O, Denmark}
\email{betermin@math.ku.dk}
\author{Mircea Petrache}
\address{Departamento de Matematicas, Pontificia Universidad Catolica de Chile, Av. Vicuna Mackenna 4860, Macul, Santiago, 6904441, Chile}
\email{mpetrache@uc.cl}

\begin{document}
\maketitle
\begin{abstract}
We investigate the minimization of the energy per point $E_f$ among $d$-dimensional Bravais lattices, depending on the choice of pairwise potential equal to a radially symmetric function $f(|x|^2)$. We formulate criteria for minimality and non-minimality of some lattices for $E_f$ at fixed scale based on the sign of the inverse Laplace transform of $f$ when $f$ is a superposition of exponentials, beyond the class of completely monotone functions. We also construct a family of non-completely monotone functions having the triangular lattice as the unique minimizer of $E_f$ at any scale. For Lennard-Jones type potentials, we reduce the minimization problem among all Bravais lattices to a minimization over the smaller space of unit-density lattices and we establish a link to the maximum kissing problem. New numerical evidence for the optimality of particular lattices for all the exponents are also given. We finally design one-well potentials $f$ such that the square lattice has lower energy $E_f$ than the triangular one. Many open questions are also presented.

\end{abstract}

\section{Introduction and main results}
 
Efforts for rigorously proving crystallization phenomena, i.e. the fact that ground states of systems exhibit a periodic order, have recently been very active. This phenomenon is observed numerically and experimentally in several settings, but its rigorous mathematical justification appears to often be very challenging and the principles at work seem to be far from being completely understood (see the reviews \cite{radin1987low, Blanc:2015yu}). In physics-inspired phenomenological models, results are known for one-dimensional models \cite{VN1,vn1d1,vn1d2,Rad1,BHS, BlancLebris,SandierSerfaty1d,Leble1d,BetKnupfNolte} and for some higher-dimensional cases \cite{Rad2,Rad3,Crystal,ELi,Stef1,Stef2,TheilFlatley,Luca:2016aa}. In parallel to this direction, the study of the problem by number theory and related combinatorial techniques (see the book \cite{ConSloanPacking}) provided important results in dimensions $2$ and also $3$ \cite{Rankin, Cassels, Diananda, Ennola, Eno2, Mont, OPS, hales2005proof, Suto1, Suto2}, as well as in some particular higher dimensions \cite{CohnElkies, banasz, Coulangeon:2010uq, CoulLazzarini, SarStromb,CoulSchurm2018}, leading to the recent proof of optimality of best packings in dimensions $8$ and $24$ \cite{Viazovska, CKMRV}.

\medskip

For $d\geq 1$, let $\mathcal L_d:=\{A\mathbb Z^d:\ A\in GL(d)\}$ be the space of all $d$-dimensional Bravais lattices. Our goal is to study here the main model problem for the crystallization question, namely the  minimization of the potential energy $E_f$ defined, for any $L\in \mathcal{L}_d$, by
\begin{equation}\label{defintro-Ef}
E_f[L]:=\sum_{p\in L\setminus\{0\}}{\vphantom{\sum}}f(|p|^2),
\end{equation}
where $|\cdot |$ is the euclidean norm on $\R^d$ and $f$ is an admissible function in the sense of Definition \ref{def:admissible} below, i.e. $f(|x|^2)$ is integrable away from the origin.

\medskip

Studying $E_f[L]$ globally amongst lattices gives an important information about the case of infinite systems of points in $\mathbb R^d$, if we answer the question of characterizing for which $f$ some special lattices (see the list at the end of this section) are, or are not, global minima amongst lattices. This type of question appeared before in several different contexts, which include Ginzburg-Landau vortices \cite{Sandier_Serfaty}, Bose-Einstein Condensates \cite{AftBN,Mueller:2002aa}, Gaussian Core models \cite{cohnkumarsoft} and Thomas-Fermi models for solids \cite{MR1842045,Betermin:2014fy} (see the review \cite{Blanc:2015yu}).

\medskip

There are two main sources of normalization which ``fix the scale'' of minimizers of $E_f$ and ensure that the minimum over $L$ (or more generally over all configurations) can be achieved: 
\begin{itemize}
\item either we \emph{fix the density of our configurations}, as a constraint in our minimization
\item or \emph{the shape of $f$ itself selects a preferred scale}, when we minimize amongst lattices of all possible scales.
\end{itemize}
For the first situation the typical example is that of Gaussian kernels $f(r^2)=e^{-a r^2}$, and in the second case the typical case is that of $f$ equal to a one-well potential such as the Lennard-Jones case $f(r^2)=a_1r^{-12}- a_2 r^{-6}$. In fact historically, the main motivation in physics for introducing one-well potentials is precisely the above, namely to provide \emph{good phenomenological models}, in which the scale-fixing is encoded directly via the potential itself, and needs not be artificially fixed (see e.g. \cite[p. 7]{Kaplan}).

\medskip

For both the above settings, the global optimality of a given lattice for $E_f$ amongst all lattices can be proved rigorously in very few examples, and the general treatment is based mainly on the following two principles:
\begin{itemize}
\item \textbf{(cristallization at fixed density)} The minimization for $f(r^2)=e^{-a r^2}$ being a Gaussian, and amongst lattices of fixed density has been treated in $d=2$ (in which case $E_f$ is the so-called \emph{lattice theta function} \eqref{def-thetalattice}), was studied in the fundamental work \cite{Mont} (see also higher-dimensional results \cite{Coulangeon:2010uq}), which prove that at fixed density the triangular lattice (defined by \eqref{def-triangularintro}) is the unique minimizer, \emph{for all choices of the variance} $a>0$. By a change of variable, this means also that for any Gaussian kernel and amongst lattices of \emph{any fixed density}, the triangular lattice is the unique minimizer. This result can be extended and transferred to all functions which are superposition of Gaussians with positive coefficients, which translates to requiring that $f$ is a \emph{completely monotone function} (see Definition \ref{def:complmonot} and the following discussion), a class which includes all inverse power admissible functions. Some conjectures from number theory are then naturally formulated for this class of $f$ (see e.g. \cite{CohnKumar}), due to the above basic principle.
\item \textbf{(minimization for one-well $f$)} In the absence of the complete monotone assumption on $f$, all known results on crystallization work under \emph{strong localization} hypotheses, and in the setting in which the study of $f$ can effectively be reduced to a \emph{finite-range situation}. The model-case to which proofs reduce is the so-called ''sticky disk potential`` in dimension $d=2$, with $f$ defined by $f(1)=-1$, $f(r)=0$ for $r>1$ and $f(r)=+\infty$ for $r<1$. In this case the general crystallization result (for the case in which we allow as competitors to $E_f$ general configurations too) can easily be proved by only discussing the nearest-neighbors of a given point, as first done by Heitmann and Radin \cite{Rad1}. The most general $f$ whose study is known to be reducible to the Heitmann-Radin situation is to our knowledge the one appearing in \cite{Crystal}, to which we refer for further references.
\end{itemize}
Note that there is a huge difference between the two above settings: any completely monotone function is in particular positive, decreasing and convex, whereas any one-well potential is negative at infinity and not monotone, and not convex. This means that \emph{a wide class of potentials does not fit in either category}, and thus escapes treatment by the known techniques.

\medskip

In this work, we extend the scope and clarify the limitations of the abovementioned Gaussian superposition and localization principles. We concentrate on the minimization amongst Bravais lattices, and in dimension $d=2$, because we feel that the main principles at work for $d=2$ are at the moment the same ones that can work also in $d\ge 3$, and there is no gain of information in treating the higher dimensions in higher generality in this work. Some of our results generalize directly to other dimensions $d\ge 3$, and we will point this out whenever this is the case. Dimension $d=1$ is better understood, but presents some important open questions. Section \ref{sec:1d} contains a survey of the main available $1$-dimensional results and examples, which are presented here as a source of inspiration for possible behaviours to test in dimension $d\ge 2$ in future work.

\medskip

As mentioned before, for the main physically relevant simple potentials, crystallization remains not rigorously proved. We mention here the following basic questions:
\begin{question}\label{q1}
If $f(r^2)$ is not a positive superposition of gaussians, can the triangular lattice still be a minimizer of $E_f$ among lattices at any fixed density?
\end{question}
\begin{question}\label{q2}
Does there exist a Lennard-Jones type potential $f(r^2)=a_1 r^{-x_1} - a_2 r^{-x_2}$ with $a_1,a_2>0, x_1>x_2>2$, such that we can prove crystallization amongst periodic configurations in dimension $d=2$? In other words, can we prove that $E_f$ has the triangular lattice as minimizer when considered on lattice configurations?
\end{question}
\begin{question}\label{q3}
Can we prove crystallization (i.e. that the minimum of $E_f$ amongst all configurations is achieved by a lattice) for some $f$ which has not extremely fast decay at infinity? For example can we prove it in dimension $d=2$ for some $f$ such that $|f(r^2)|\ge C r^{-6}$ for all large enough $r$?
\end{question}
Our main results are as follows:
\begin{enumerate} \item For the long-range case, we formulate criteria depending on the expression of $f$ as a superposition of exponentials, and we distinguish different behaviours based on the sign of the inverse Laplace transform of $f$ in Theorem \ref{MainTh1}. In particular we provide an example in Section \ref{not_cm} which gives a positive answer to Question \ref{q1}.
\item We prove, in Proposition \ref{Mainprop1} and Theorem \ref{MainTh4} that for a large class of one-well potentials \emph{the square lattice has lower energy than the triangular one}, which proves for the first time in dimension $d>1$ that the two frameworks described above can have essentially distinct properties.
\item In case of one-well potentials, in Theorem \ref{MainTh2} and \ref{MainTh3}, we reduce the number of parameters needed to understand the behavior of $f$ belonging to the class of Lennard-Jones type potentials (this larger class was originally introduced by Mie in 1903 \cite{mie1903kinetischen}), and we give new evidence for a positive answer to a \emph{stronger version} of Question \ref{q2}, which however we do not have the tools to answer.
\end{enumerate}
The proofs of our results and counterexamples are implemented by extracting precise principles which could also be applied in dimension higher than $2$.

\medskip

We have already studied the above kind of questions for some more specific choices of $f$ and in some related problems, in \cite{Betermin:2014fy,BetTheta15,Beterloc,BeterminPetrache,Beterminlocal3d,MorseLB}, again with special emphasis on ``physical'' dimensions $d\in \{2,3\}$. The first author and Kn\"upfer have treated the case of two-dimensional interaction of masses located on lattice sites in \cite{BetKnupfdiffuse,Softtheta}, and the second author and Serfaty have treated the case of Jellium-type energies in $2$-dimensions for power-law interactions in \cite{petrache2017next}.

\medskip

Some noteworthy Bravais lattices, which play important roles in our energy minimization problem, are the square lattice $\mathbb Z^2$ and the triangular lattice $\mathsf{A}_2$ and its renormalized version $\Lambda_1$, defined by
\begin{equation}\label{def-triangularintro}
\mathsf{A}_2:=\Z(1,0)\oplus \Z\left(\frac{1}{2},\frac{\sqrt{3}}{2}  \right),\quad \Lambda_1:=\sqrt{\frac{2}{\sqrt{3}}}\mathsf{A}_2.
\end{equation}
In dimension $3$, by abuse of notation (because the lattices $\mathsf{D}_n$ are not usually with a unit density), special roles will be played by the Face-Centred-Cubic (FCC) lattice $\mathsf{D}_3$ and its dual, the Body-Centred-Cubic (BCC) lattice $\mathsf{D}_3^*$, which will be defined by
\begin{align}
&\mathsf{D}_3:=2^{-\frac{1}{3}}\left[ \Z(0,1,1)\oplus \Z(1,0,1) \oplus \Z(1,1,0)  \right],\label{defFCC}\\ 
&\mathsf{D}_3^*:=2^{\frac{1}{3}}\left[\Z(1,0,0)\oplus \Z(0,1,0)\oplus \Z\left(\frac{1}{2},\frac{1}{2},\frac{1}{2}  \right)  \right].\label{defBCC}
\end{align}
We have represented these lattices in Figure \ref{Lattices}.
\begin{figure}[!h]
\centering
\includegraphics[width=6cm]{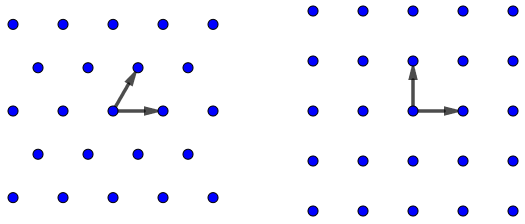} \\
\includegraphics[width=3cm]{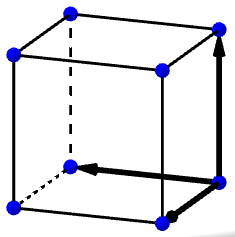}\quad\includegraphics[width=3cm]{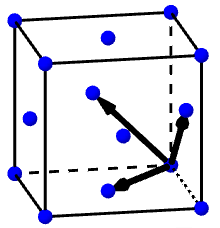}\quad\includegraphics[width=3cm]{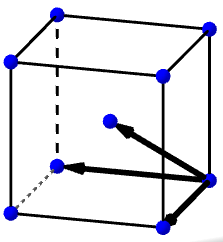}
\caption{Representation of the triangular and square lattices $\mathsf{A}_2,\Z^2$ (first line) and the simple cubic, FCC and BCC lattices $\Z^3,\mathsf{D}_3,\mathsf{D}_3^*$ (second line)}
\label{Lattices}
\end{figure}

In higher dimensions $d\in \{4,8,24\}$, we consider (unit-density versions of) the classical lattices $\mathsf{D}_4,\mathsf{E}_8,\Lambda_{24}$, which are respectively defined in \cite[Sect. 7.2, 8.1 and 11]{ConSloanPacking}.

\medskip

We now pass to introducing the precise statements of our results.

\subsection{Minimization at fixed density by writing $f(r^2)$ as a superposition of Gaussians}\label{sec:intro_gauss}

We recall the notion of completely monotone function:
\begin{definition}\label{def:complmonot}
We say that $f:(0,+\infty)\to [0,+\infty)$ is \emph{completely monotone} if for any $k\in \N_0$ and any $r>0$ $(-1)^{k}f^{(k)}(r)\geq 0$. 
\end{definition}

We introduce the following normalizations and notations on spaces of lattices. Denote respectively by $\mathcal{L}_d^\circ\subset \mathcal L_d $ and $\mathcal{L}_d^1\subset \mathcal L_d$ the subsets of lattices with respectively unit density and unit shortest non-zero vector. Furthermore, we denote by $D_{\mathcal{L}_d}$ and $D_{\mathcal{L}_d^\circ}$ the fundamental domains of $\mathcal L_d$ and $\mathcal L_d^\circ$ (see Section \ref{sect-notations} for a precise definition), where each Bravais lattice appears only once. Moreover, the shape $[L]$ of a Bravais lattice $L\in \mathcal{L}_d^\circ$ is its equivalence class modulo rotation and dilation among Bravais lattices (see Definition \ref{defn:shape}).

\medskip

Our starting point is the well-known result (see e.g. \cite[p. 169]{CohnKumar} or \cite[Prop 3.1]{BetTheta15}) which says that if $f$ is a completely monotone (admissible) function, then the triangular lattice $\Lambda_1$ is the unique minimizer of $L\mapsto E_f[\lambda L]$ on $D_{\mathcal L_d^\circ}$ for any fixed $\lambda>0$. Indeed, this follows by superposition, from the following two celebrated results. 
\begin{theorem*}[{Montgomery's Theorem \cite[Thm 1]{Mont}}] In dimension $d=2$, the triangular lattice $\Lambda_1$ is the unique minimum of $D_{\mathcal{L}_2^\circ}\ni L\mapsto \theta_L(\alpha)$, for any fixed $\alpha>0$.
\end{theorem*}
The lattice theta function $\theta_L$ mentioned above is defined for $L\in \mathcal{L}_d$ (or, more generally, for any configuration $L\subset \mathbb R^d$ for which the below sum is finite) and $\alpha>0$ by
\begin{equation}\label{def-thetalattice}
\theta_L(\alpha):=\sum_{p\in L} e^{-\pi \alpha |p|^2}.
\end{equation}
\begin{theorem*}[Hausdorff-Bernstein-Widder Theorem \cite{Bernstein}] A function $f$ is completely monotone if and only if $f$ is the Laplace transform of a positive Borel measure $\rho_f$ on $(0,+\infty)$.
\end{theorem*}
A direct important consequence, also proved by Rankin \cite{Rankin}, Ennola \cite{Eno2}, Cassels \cite{Cassels} and Diananda \cite{Diananda}, is the minimality of $\Lambda_1$ on $D_{\mathcal{L}_2^\circ}$, for any $s>2$, for the Epstein zeta function defined by
\begin{equation}\label{def-epsteinzeta}
\zeta_L(s)=\sum_{p\in L\setminus\{0\}} \frac{1}{|p|^s}.
\end{equation}
The next natural question is now to study the minimization of $L\mapsto E_f[\lambda L]$ on $D_{\mathcal{L}_d^\circ}$, for fixed $\lambda>0$, when instead of being positive as in the above theorem, the measure $\mu_f$ is negative on some open sets of $(0,+\infty)$, i.e. $f$ is not completely monotone. This question has been studied by the first author and Zhang in \cite{Betermin:2014fy,BetTheta15} for the special case of the Lennard-Jones type potentials \eqref{defLJ} in dimension $d=2$, where the minimality of $\Lambda_1$ for $\lambda$ small enough was proved (with an explicit upper bound), as well as its non-minimality for $\lambda$ large enough (again with an explicit lower bound). Another natural question is the nature of lattices $\Lambda$ that are minimizers of $L\mapsto E_f[\lambda L]$ at fixed density $\lambda$ for any $\lambda>0$. Using the Laplace transform representation
\begin{equation}\label{laplacerepresent}
\forall r>0,\quad f(r)=\int_0^{+\infty} e^{-rt}d\mu_f(t),
\end{equation}
where $\mu_f=\mathcal{L}^{-1}[f]$ is the inverse Laplace transform of $f$, which we assume to be well defined and to be a Radon measure, we prove the following results in Proposition \ref{prop-asymptnonopt} and Proposition \ref{prop-anyscale2} below.
\begin{theorem}[Minimization at fixed scale]\label{MainTh1}
Let $d\geq 1$ and assume that $L_m$ is, for all $\alpha>0$, the unique minimizer on $D_{\mathcal{L}_d^\circ}$ of $L\mapsto \theta_L(\alpha)$, defined by \eqref{def-thetalattice}. Let $f$ be an admissible potential with representation \eqref{laplacerepresent}. Then the following holds for any $r_0>0$:
\begin{enumerate}
\item If $\mu_f$ is negative on $(0,r_0)$, then for any Bravais lattice $L\in \mathcal{L}_d^\circ\backslash\{L_m\}$, there exists $\lambda_0$ such that for any $\lambda>\lambda_0$, $E_f[\lambda L]<E_f[\lambda L_m]$;
\item  If $\mu_f$ is negative on $(r_0,+\infty)$, then for any Bravais lattice $L\in \mathcal{L}_d^\circ\backslash\{L_m\}$, there exists $\lambda_1$ such that for any $0<\lambda<\lambda_1$, $E_f[\lambda L]<E_f[\lambda L_m]$.
\item If $\mu_f$ is positive on $(0,r_0)$ or on $(r_0,+\infty)$, and $\Lambda$ is a minimizer of $L\mapsto E_f[\lambda L]$ for any $\lambda >0$ on $D_{\mathcal{L}_d^\circ}$, then $\Lambda=L_m$.
\item If $\mu(r)$ is negative on $(0,r_0)$ or on $(r_0,+\infty)$, then the minimizer of $L\mapsto E_f[\lambda L]$ cannot be the same for all $\lambda>0$.
\end{enumerate}
In particular, these results hold in dimension $d=2$ for the triangular lattice $L_m=\Lambda_1$.
\end{theorem}

The major issue is to identify $L_m$ for a given dimension $d$. This result is only known, so far, in dimension $d=2$ where $L_m=\Lambda_1$. We have previously worked on this problem in higher dimensions and we have given in \cite{BeterminPetrache} many results about minimizers of $L\mapsto \theta_L(\alpha)$ on different subclasses of lattices. It is actually conjectured by Cohn and Kumar in \cite[Conjecture 9.4]{CohnKumar} that $\Lambda_1,\mathsf{E}_8$ and $\Lambda_{24}$ are the unique minimizers of $L\mapsto \theta_L(\alpha)$ on $D_{\mathcal{L}_d^\circ}$, $d$ respectively equal to $\{2,8,24\}$ for any fixed $\alpha>0$. Their conjecture is even more general than that: they claim that these lattices, as well as $\Lambda_1$ in dimension $d=2$, are the unique minimizers  of $\mathcal C\mapsto\theta_{\mathcal C}(\alpha)$, for all $\alpha>0$, among \emph{all periodic configurations $\mathcal C$} of density $1$. For a space $X$ (in our case we always assume $X=\mathbb R^d$), configurations in $X$ with the property of minimizing the energy $E_f$, at fixed density, for all Gaussian kernels $f(r^2)=e^{-ar^2}$, are called \emph{universally optimal}. The local minimality of $\mathsf{D}_4$ among four-dimensional periodic configurations of unit density for the lattice theta function, for all $\alpha>0$, proved in \cite{Coulangeon:2010uq}, suggests that $\mathsf{D}_4$ should also be universally optimal in dimension $d=4$. Furthermore, recent results by Viazovska et al. \cite{Viazovska,CKMRV} about the best packings in dimensions $8$ and $24$ have shown the efficiency of the Cohn-Elkies linear programming bounds for sphere packing \cite{CohnElkies} and could possibly lead to a proof of this conjecture in those dimensions (see \cite{CohnKumar} for the link between linear programming bounds and energy minimization problems).

\medskip

A natural conjecture would be that for an admissible function $f$, if for any $\lambda>0$, the minimizer $L_m$ of the theta function for any $\alpha>0$ is the unique minimizer of $L\mapsto E_f[\lambda L]$ on $D_{\mathcal L_d^\circ}$, then $f$ is completely monotone (see e.g. \cite{BetTheta15}), i.e. $\mu_f$ is a positive measure. Our example of Section~\ref{not_cm} shows that this is not true in dimension $d=2$ for $L_m=\Lambda_1$. Indeed, considering, for $\varepsilon\geq 0$, the following potential
\begin{equation}\label{deffepsilon}
f_\varepsilon(r):=\frac{6}{r^4}-\frac{2(2+\varepsilon)}{r^3}+\frac{1+\varepsilon}{r^2},
\end{equation}
such that $\mu_{f_\varepsilon}$ is negative on $(1,1+\varepsilon)$, we numerically show that there exists $\varepsilon_0\approx 1.148$ such that for any $0< \varepsilon<\varepsilon_0$, the triangular lattice $\Lambda_1$ is the unique minimizer of $L\mapsto E_{f_\varepsilon}[\lambda L]$ on $D_{\mathcal{L}_d^\circ}$ for any fixed $\lambda>0$. We therefore observe that the fact that $\mathcal{L}^{-1}f$ is negative in a small interval, and if the measure of this negative part is not too big with respect to the measure of its positive part, the triangular lattice stays optimal at all scales for $E_f$. This results leads to another one that is more general, concerning the measure of the negative and positive parts of $f$:
\begin{question}``How negative'' can the inverse Laplace transform $\mu_f:=\mathcal L^{-1}f$ be, while preserving the property that the minimum $\min_{L\in D_{\mathcal L_2^\circ}}E_f[\lambda L]$ is achieved at all $\lambda>0$ by the triangular lattice? For example, if $\mu_f=\mu_f^+-\mu_f^-$ with $\mu_f^\pm$ positive finite measures, then we can ask more precisely: how large can the ratio of $R_f:=\int d\mu_f^-/\int d\mu_f^+$ be while preserving the above property?
\end{question}

An answer to the above question would be of a great interest for understanding the following problem that we leave open:
\begin{question}
What is the largest class of functions $f$ such that for any $\lambda>0$, the triangular lattice $\Lambda_1$ is the unique minimizer of $L\mapsto E_f[\lambda L]$ on $D_{\mathcal L_2^\circ}$?
\end{question}
Note that a similar question arises also in general dimension $d\ge 2$, and is also open in that case. Regarding the case $d=1$, we also do not know the complete answer, but more results are available. See Section \ref{sec:1d} for a discussion.

\medskip

Concerning our example \eqref{deffepsilon}, it is also natural to ask whether the triangular lattice is a minimizer, for the same potential, also amongst general configurations. More generally, the following question is completely open:
\begin{question}
Is there any non-completely monotone $f$ for which the minimizer of $E_f[\lambda \mathcal{C}]$ is the triangular lattice for all $\lambda>0$, among periodic configurations $\mathcal{C}$ of unit density?
\end{question}
We note again, by Theorem \ref{MainTh1}, that one-well potentials $f$ are not a good candidate in the above question, and we do not know of a good underlying principle which would allow to find a candidate. It is for instance not clear if $f_\varepsilon$ defined by \eqref{deffepsilon} could satisfy this property. However, we conjecture that the local minimality of $\Lambda_1$ for $E_{f_\varepsilon}$ should hold among periodic configurations of fixed density by a small modification of \cite[Cor. 4.5]{Coulangeon:2010uq}.

\medskip

Also note that, while no proof of any lattice-like configuration being \emph{universally optimal} is available in the literature, our proof in Section \ref{not_cm} shows that the property of a triangular lattice to be a \emph{minimizer of the energy amongst lattices at any fixed scale}, usually conjectured in universal optimality frameworks, \emph{could hold beyond that setting}, i.e. for non-completely monotone functions. Therefore an answer to the above question would be important for understanding/testing the relevance of the concept of universal optimality, and we leave the following general question open:

\begin{question}[Class of functions satisfying universal optimality] What is the largest class of functions $f$ such that the unique minimizer $L_m$ of $L\mapsto \theta_L(\alpha)$ on $D_{\mathcal{L}_d^\circ}$, for any $\alpha>0$ is the unique minimizer of $L\mapsto E_f[\lambda \mathcal{C}]$, for any $\lambda>0$, among periodic configurations $\mathcal{C}$ of unit density?

\end{question}

\medskip

Furthermore, about the minimality of some lattice at all scales, the following question can be asked:
\begin{question}\label{q3d} For $d=3$ does there exist a continuous potential $f$ and a lattice $L_f$ for which $\min_{L\in D_{\mathcal L_3^\circ}}E_f[\lambda L]$ is achieved by $L_f$ at all scales $\lambda>0$?
\end{question}
The reason why the above is not clear is that in $d=3$ the analogue of the Rankin-Montgomery theorem \cite{Rankin, Mont}, is false, as noted for example in \cite[p. 117]{SarStromb}, i.e. no lattice can achieve the minimum for the theta functions at all scales. This implies that if a function $f$ as required in Question~\ref{q3d} exists, then it cannot have positive inverse Laplace transform. On the other hand, for the best packing problem, Hales' result \cite{hales2005proof} implies that the unique minimizer amongst lattices is given by the FCC lattice at all scales, and it is well known that this minimization problem is the limit $s\to+\infty$ of the minimization for $f(r)=r^{-s}$.

\subsection{Minimization for one-well potentials $f$ without density constraint}

The second problem we are setting in this paper is the global optimality of some lattices for $E_f$ on $D_{\mathcal{L}_d}$ (without a density restriction) where $f$ is a one-well potentials, i.e. $f:(0,+\infty)\to \R\cup \{+\infty\}$ that is decreasing on $(0,a)$ and increasing on $(a,+\infty)$ for some $a>0$. In \cite{BetTheta15,Crystal}, three examples of one-well potentials have been studied, where the global  optimality of a triangular lattice was proved in dimension $2$ (on $D_{\mathcal{L}_2^\circ}$ for the two first and among all configurations, in the thermodynamic limit sense, for the third one). Those are:
\begin{enumerate}
\item Lennard-Jones type potentials $f_{\vec a,\vec x}^{LJ}$ with parameters $\vec a=(a_1,a_2)\in (0,+\infty)^2$ and $\vec x=(x_1,x_2)$, $x_2>x_1>2$, such that, for any $r>0$,
\begin{equation}
f_{\vec a,\vec x}^{LJ}(r)=\frac{a_2}{r^{x_2/2}}-\frac{a_1}{r^{x_1/2}} \quad \textnormal{with}\quad \pi^{-\frac{x_2}{2}}\Gamma\left(\frac{x_2}{2}\right)\frac{x_2}{2}\leq \pi^{-\frac{x_1}{2}}\Gamma\left(\frac{x_1}{2}\right)\frac{x_1}{2}.
\end{equation}
The classical Lennard-Jones potential is $r\mapsto f_{\vec a,\vec x}^{LJ}(r^2)$ with $x_1=6, x_2=12$ in our notation. The exponent $x_1=6$ is justified equals the long-range behavior of the Van der Waals interaction (see e.g. \cite[p. 10]{Kaplan}), whereas we do not know a good physical intuition behind the choice $x_2=12$. Note that $f_{\vec a,\vec x}^{LJ}$ is admissible in dimension $d$ according to our definition, if and only if $x_1>d$.
\item Differences of (three-dimensional) Yukawa potentials $f_{\vec a,\vec x}^Y$ with parameters $\vec a=(a_1,a_2)$, $0<a_1<a_2$ and $\vec x=(x_1,x_2)$, $0<x_1<x_2$, such that, for any $r>0$,
\begin{equation}
f_{\vec a, \vec x}^Y(r)=\frac{a_2 e^{-x_2 r}-a_1 e^{-x_1 r}}{r} \quad \textnormal{with} \quad \frac{a_1\left( a_1x_2+x_1(a_2-a_1)\pi \right)}{a_2x_2\left(a_1+(a_2-a_1)\pi\right)}e^{\left(1-\frac{x_1}{x_2}  \right)\left( \frac{a_2}{a_1}-1 \right)\pi}\geq 1.
\end{equation}
This type of interacting potential arises in physics. For example, it turns out that Neumann \cite{Neumann} has shown that a linear combinations of Yukawa potentials are the most general laws ensuring the stability of electric charges.
\item Abstract potentials $V_\alpha$ described in \cite{Crystal}, such that, for $\alpha$ small enough, there is a large repulsion at $0$ of order $1/\alpha$, a well with width of order $\alpha$ and behavior at infinity controlled by $r\mapsto \alpha r^{-7}$. 
\end{enumerate}
We first focus on Lennard-Jones type potentials and we will write, in the following two results,
\begin{equation}\label{defLJ}
f(r)=f_{\vec a,\vec x}^{LJ}(r)=\frac{a_2}{r^{x_2/2}}-\frac{a_1}{r^{x_1/2}}, \quad x_2>x_1>d,\quad (a_1,a_2)\in (0,+\infty).
\end{equation}
In this case, the lattice energy of a Bravais lattice $L\in \mathcal{L}_d$ is given by
\begin{equation}
E_f[L]=a_2\zeta_L(x_2)-a_1\zeta_L(x_1),
\end{equation}
where the Epstein zeta function $\zeta_L$ is defined by \eqref{def-epsteinzeta}.

\medskip

These potentials arise naturally in physics models of matter (see e.g. \cite{mie1903kinetischen,LJonesPhasediagram,Kaplan,LJonesnoblegas}) in the case of the Born-Oppenheimer adiabatic approximation of the interaction energy where the electrons effect is neglected and the energy is reduced to the atomic interaction of the nuclei (see e.g. \cite[p. 33]{CondensMatter}). Thus the potential energy of the system is expressed in terms of many-body interactions and the simplest case -- which is also relevant in many situations (see e.g. \cite[p. 945]{CondensMatter}) -- is the one in which the energy is a sum of $2$-body interaction potentials. Lennard-Jones type potentials also appear in social aggregation model \cite{MEKBS}.

\medskip

In dimension $d=2$ (resp. $d=3$), the global minimizer of $E_f$ on $D_{\mathcal{L}_d}$ is expected to be a triangular lattice (resp. a FCC lattice) for all $x_2>x_1>d$, as conjectured in \cite{BetTheta15,Beterloc,Beterminlocal3d} from numerical evidences and local optimality results. In the following result, proved in Proposition \ref{prop-LJglobal}.(1) and Proposition \ref{prop:equiv} below, we show that the problem of minimizing $E_f$ on the space of all Bravais lattices $D_{\mathcal{L}_d}$ can be reduced to a minimization problem on the space of lattices with unit density $D_{\mathcal{L}_d^\circ}$. In particular, the shape of the global minimizer of $E_f$ does not depend on $a_1,a_2$, which we have already observed in \cite{BetTheta15}. Furthermore, inspired by the $d\in \{2,3\}$ cases where the minimizer of $E_f$ on $D_{\mathcal{L}_d}$ seems to be the same for all $x_2>x_1>d$, we give different statements that are equivalent to this optimality for all the parameters $x_1,x_2$.

\begin{theorem}[Minimality for the Lennard-Jones energy]\label{MainTh2}
Let $d<x_1<x_2$, $(a_1,a_2)\in (0,+\infty)^2$ and $L_0=\lambda \Lambda$ where $\Lambda \in \mathcal{L}_d^\circ$. Then for $f$ a Lennard-Jones type potential as in \eqref{defLJ}, $L_0$ is a global minimizer of $E_f$ on $D_{\mathcal{L}_d}$ if and only if $\Lambda $ is a minimizer on $D_{\mathcal{L}_d^\circ}$ of $\tilde{E}_f$ defined by
\begin{equation}
\tilde{E}_f[L]:=\frac{\zeta_L(x_2)^{x_1}}{\zeta_L(x_1)^{x_2}}.
\end{equation}

Furthermore, if $L\in \mathcal L_d^\circ$ then we define functions $H_L,h_L:(d,+\infty)\to \R$ by
$$
H_L(x):=\frac{1}{x}\log\left(  \frac{\zeta_L(x)}{\zeta_{\Lambda }(x)}\right), \quad
h_L(x):=-\log \zeta_L(x) +x\frac{\partial_x\zeta_L(x)}{\zeta_L(x)}.
$$
The following conditions are equivalent:
\begin{enumerate}
\item For any $x_2>x_1>d$, the lattice $L_0 $ is the unique minimizer of $E_f$ on $D_{\mathcal L_d}$.
\item  For any $x_2>x_1>d$, the lattice $\Lambda $ is the unique minimizer of $\tilde{E}_f$ on $D_{\mathcal L_2^\circ}$.
\item For any Bravais lattice $L\in \mathcal L_d^\circ\setminus\{ \Lambda \}$, the function $H_L$ is strictly increasing on $(d,+\infty)$.
\item For any $x>d$, $\Lambda $ is the unique minimizer of $h_L(x)$ on $D_{\mathcal{L}_d^\circ}$.
\end{enumerate}
\end{theorem}
\begin{rmk}
About the case $d=2$, (1) of Theorem~\ref{MainTh2} has been proved in \cite[Thm. 1.2.B.2 and Lem. 6.17]{BetTheta15} on the small interval $I=(2,2\psi^{-1}(\log \pi)-2)$ where $\psi$ is the digamma function and $2\psi^{-1}(\log \pi)-2\approx 5.256$. 
\end{rmk}
We conjecture that the equivalent statements of Theorem \ref{MainTh2} hold true in dimensions $d\in \{ 2,3,4,8,24\}$ for $\Lambda\in \{\Lambda_1,\mathsf{D}_3,\mathsf{D}_4,\mathsf{E}_8,\Lambda_{24}\}$. In dimensions $d\in \{2,3\}$, new numerical evidence for $\tilde{E}_f$ supporting our conjecture is included in Figures \ref{H2to10square}, \ref{fig:ratio}, \ref{Z3vsFCCBCC} and \ref{fig-BCCFCC}. Next, we prove the global optimality of these lattices for large values of the parameters and find their asymptotically optimal scaling as $x_1,x_2\to +\infty$. For a definition and some known results on the \emph{kissing number} (also called \emph{coordination number} in crystallography) of a Bravais lattice $L\in\mathcal L_d$, denoted $\tau(L)$, see Definition \ref{def:kissing} and Remark \ref{rmk:kissing} below. The following result is proved in Proposition \ref{prop-LJglobal}.(2) and Proposition \ref{prop:asymin} below.

\begin{theorem}[Global minimizer of the Lennard-Jones energy for large parameters]\label{MainTh3}
For $L,\Lambda\in\mathcal L_d$ and $f$ a Lennard-Jones type potential as in \eqref{defLJ} there holds 
\begin{equation}\label{limit-kissing}
\lim_{x_1\to +\infty}\lim_{x_2\to +\infty}\frac{\displaystyle \min_{\lambda>0}E_f[\lambda\Lambda]}{\displaystyle \min_{\lambda>0}E_f[\lambda L]}=\frac{\tau(\Lambda)}{\tau(L)}\ .
\end{equation}
As the above minimum values are negative, if $\tau(\Lambda)$ uniquely realizes the optimal kissing number $k(d)$ amongst lattices, then there exists $x_0=x_0(d)$ such that for any $x_2>x_1>x_0$, the unique minimizer of $E_f$ on $D_{\mathcal{L}_d}$ has shape $[\Lambda]$. In particular, this holds for the cases $(d,\Lambda)\in\{(2,\mathsf{A}_2), (3,\mathsf{D}_3), (4,\mathsf{D}_4),(8,\mathsf{E}_8),(24,\Lambda_{24})\}$.

\medskip

If the minimum of $f(r)$ is achieved at $r_f=r_f(x_1,x_2,a_1,a_2)>0$, and $\lambda^{L_2}_0=\lambda^{L_2}_0(x_1,x_2,a_1,a_2)>0$ is the factor such that $\lambda^{L_2}_0 L_2$, $L_2\in \mathcal{L}_d^1$, realizes the minimum energy among lattices of the same shape as $L_2$, i.e. $\min_{\lambda>0} E_f[\lambda L_2]=E_f[\lambda^{L_2}_0 L_2]$, then the following limit exists and satisfies
\begin{equation}\label{limit_r0intro}
\lim_{\substack{x_1,x_2\to +\infty\\x_1<x_2, r_f=r_0}}\lambda^{L_2}_0(x_1,x_2,a_1,a_2)=\sqrt{r_0}.
\end{equation}
\end{theorem}

This result goes in a similar direction as Theil's work \cite{Crystal}, valid for general (not necessarily lattice-like) configurations: increasing the parameters makes at the same time the repulsion near the origin larger, the decay at infinity faster and the well of $f$ narrower. Thus, in the limit, the energy takes only into account the nearest-neighbours, that is why the lattices achieving the optimal kissing number $k(d)$ are globally optimal and the length of these lattices tend to the square root of the value of the minimizer of $f$ as in \eqref{limit_r0intro}.
\begin{rmk}
In the three-dimensional case, the phase diagram of the classical Lennard-Jones energy $f(r^2)=r^{-12}-r^{-6}$ has been numerically studied in detail in \cite{LJonesPhasediagram,StillingerLJ}. In particular, by \cite{StillingerLJ}, the optimizer of $E_f$ amongst all configurations seems to be the hexagonal close-packing (HCP, see \cite[Sect. 6.5]{ConSloanPacking} for a precise definition). The minimization of $E_f$ corresponds to the minimization of the Helmholtz free energy $H=U-TS$ in the regime of small temperature and pressure $T,P\approx 0$, where $U$ is the internal energy and $S$ is the entropy of the system. On the other hand, for high pressures, numerically, it seems that the global minimizer is the FCC lattice. The HCP is not a Bravais lattice but is also known to be a best packing in dimension $d=3$ (see \cite{hales2005proof}) as well as the solid structure of many chemical components. As the number of closest neighbors of any point is the same for HCP and FCC, \eqref{limit-kissing} also gives the optimality of the HCP among all periodic configuration for sufficiently large parameters, as it is also true for the FCC lattice, but distinguishing the difference of energies of these two best packings.
\end{rmk}

Note that if we try to fix a scale constraint while minimizing $E_f$ for one-well potentials, then in general \emph{we will find different minimizers at different scales}. The \emph{local minima} for this problem have been numerically studied in \cite{Beterloc}, proving that, as $\lambda>0$ grows, it is expected that the minimizer changes from a triangular lattice to a rhombic one, then to a square one, then to a rectangular lattice and then to a degenerate rectangular one, and parts of this result are actually rigorously proved for high and low densities in \cite{BetTheta15,Beterloc}.

\medskip

It is therefore interesting to understand numerically and then to prove rigorously what are possible behaviours of the global minimizers e.g. of simple/explicit one-well physically inspired potentials at fixed scale, and what are the principles that govern this behavior:
\begin{question}[Dimension $d=3$ phase diagram]
What is the phase diagram with respect to $\lambda>0$ of $\min_{L\in D_{\mathcal L^\circ_3}} E_f[\lambda L]$, where $f$ is the classical Lennard-Jones potential $f(r^2)=r^{-12}-r^{-6}$? I.e., how are the minimizing lattices at fixed density varying, as $\lambda$ increases?
\end{question}

It would be interesting to know whether the minimizer of $E_f[\lambda L]$ on $D_{\mathcal{L}_2^\circ}$ changes with $\lambda$, similarly to the classical Lennard-Jones case described in \cite{Beterloc}, but for the more general case in which $f$ is the difference of two completely monotone functions (see also \cite{radinclassground} for a study of such $f$ under extra conditions in dimension $d=1$). More precisely, we ask the following question:

\begin{question}
Let $g_1,g_2$ be completely monotone functions and $f$ be defined by $f(r):=g_1(r)-g_2(r)$. Is the phase diagram, with respect to $\lambda>0$, of $\min_{L\in D_{\mathcal{L}_3^\circ}}E_f[\lambda L]$ the same as the one for the classical Lennard-Jones potential $f(r)=r^{-12}-r^{-6}$ described in \cite{Beterloc}, i.e. triangular-rhombic-square-rectangular-degenerate as $\lambda$ increases?
\end{question}

\subsection{Constructions of $f$ which favor the square lattice over the triangular one}

Since the shape of the global minimizer of $E_f$ does not depend on $a_1,a_2$, it turns out that it is possible to construct potential with a well as wide or narrow as we want. Hence, the existence of a one-well potential $f$ such that the global minimizer is not triangular appears as an interesting question. Inspired by Ventevogel's counter-example \cite[Sec. 5]{VN1} given in the one-dimensional setting, we give in Section \ref{subsec-onewell} an example of discontinuous potential (see Figure \ref{CE1}) and an example of continuous potential such that
\begin{equation}\label{min-lambda}
\displaystyle \min_{\lambda>0} E_f\left[\lambda \Z^2\right]=E_f\left[\lambda_0^{\mathbb Z^2} \Z^2\right]<E_f\left[\lambda_0^{\mathsf{A}_2} \mathsf{A}_2\right]=\min_{\lambda>0} E_f[\displaystyle \lambda \mathsf{A}_2],
\end{equation}
where $f$ is defined by $f(r^2)=g(r)$ (change of notation justified by the fact that $g$ is the potential whose derivative is estimated over the lattices distances in the proof). More precisely, the proof of the following proposition is given in Appendix \ref{appendix}.

\begin{proposition}[One-well potential such that the global minimizer is not triangular - explicit example]\label{Mainprop1}
Let $g$ be the continuous potential defined by
$$
g(r):= \left\{ \begin{array}{ll}
\displaystyle \frac{\left( \frac{2}{3} \right)\left(\frac{4}{9}  \right)^{p}}{r^{p}} & \mbox{if $0<r<4/9$}\\
2-3 r & \mbox{if $4/9\leq r\leq 1$}\\
-r^{-4} & \mbox{if $r>1$.}\\
\end{array}\right.
$$
Then there exists $p_0$ such that for any $p>p_0$, \eqref{min-lambda} holds for $f$ defined by $f(r^2):=g(r)$.
\end{proposition}

\begin{figure}[!h]
\centering
\includegraphics[width=9cm]{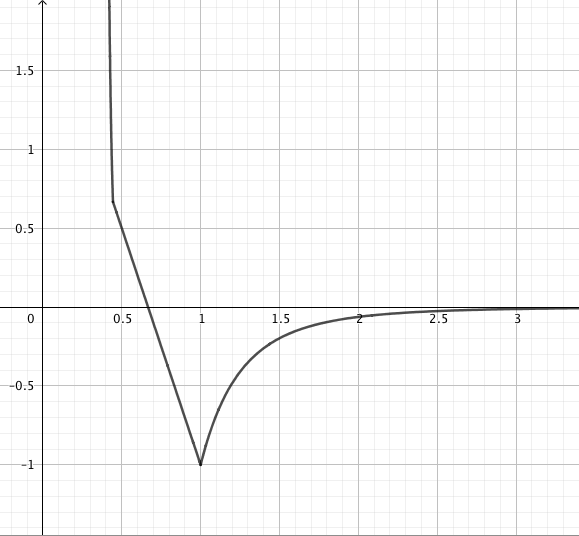} \\
\caption{Plot of function $f$ defined in Proposition \ref{Mainprop1}}
\label{CE1}
\end{figure}
It is also possible to construct a large class of $C^1$-functions $f$ having property \eqref{min-lambda}. More precisely, we have the following result (see Proposition \ref{prop:theil_ZA} for a statement giving a more detailed descriptions of the functions $g$ below).
\begin{theorem}[One-well potential such that the global minimizer is not triangular - general result]\label{MainTh4}
There exists an uncountable class of functions $g:[0,+\infty)\to\R\cup\{+\infty\}$ such that the function defined by $f(r^2)=g(r)$ is admissible, and \eqref{min-lambda} holds.
\end{theorem}

Here we have designed potentials that favour a square lattice instead of a triangular lattice. These two results are in the same spirit as the work of Torquato et al. \cite{TorquatoHCb,Torquato09,Torquato1} where a (truncated) potential is designed in such a way that a targeted lattice structure is the minimizer of the interaction energy. Our potential is however not truncated and the proof of Theorem \ref{MainTh4} can be generalized to another $d$-dimensional lattices (see Remark \ref{rmk-generalprinciple}). However, there is nothing in Theorem \ref{MainTh4} that shows what is the global minimum of $E_f$ (it is not necessarily a square lattice).

\medskip

Finally, we mention the following open direction:
\begin{open}[Stability of crystallization phenomena, with respect to perturbations of $f$]
Study and classify natural distances between (or other measures of the size of perturbations of) interaction kernels $f$, with respect to which small perturbations of $f$ can be ensured to preserve the crystallization properties of the kernels, such as the existence and shape of the global minimum amongst periodic configurations.
\end{open}
In fact our counterexample in Section \ref{not_cm} of a non-completely monotone $f$ which is a good candidate for crystallization at all scales is based on a \emph{small-perturbation method}, and the same can be said more generally for the many of the available proofs of crystallization. Thus it seems extremely important to understand what is the right notion of stability of potentials $f$, as hinted at above, even if the known cases are at the moment extremely episodic and prevent us from formulating any more precise question in a compelling way. 

\medskip

We mention, as a simple possible starting point for such studies in dimension $d=1$, the fact that in \cite{hamrickradin}, a potential which is a small (in $C^1$-norm but not in $C^2$-norm) perturbation of a one-well potential was produced, for which $N$-point minimizers converge to a \emph{quasicrystal}. The principle in that case is to ``add small wells'' to $f$ an irregular pattern: this creates the possibility to find, next to the standard periodic configuration, a slightly better (in terms of energy), but non-periodic, ground state.

\subsection*{Plan of the paper.} In Section \ref{sect-notations} we give the precise definitions of the objects we are working with. A survey of one-dimensional results is presented in Section \ref{sec:1d}. Section \ref{sect-allscales} is devoted to the problem of minimizing $E_f$ at fixed scale. In particular, we prove Theorem \ref{MainTh1} in this section. The Lennard-Jones type case is investigated in Section \ref{sectLJ}, where Theorem \ref{MainTh2} and \ref{MainTh3} are proved. Finally, in Section \ref{sect-contrex}, several counter-examples are stated and proved, included Theorem \ref{MainTh4}. Proposition \ref{Mainprop1} is proved in Appendix \ref{appendix}.

\section{Notations and definitions}\label{sect-notations}
Let $\mathcal L_d:=\{A\mathbb Z^d:\ A\in GL(d)\}$ be the space of all $d$-dimensional Bravais lattices and $\mathcal{P}_d:=\{M\in \mathbb R^{d\times d}:\ M^t=M, M\mbox{ positive definite}\}$ the cone of positive definite matrices, in turn identified with the cone of positive definite quadratic forms in $d$ (real) variables. Recall that the link between the above settings is the following: to each $A\mathbb Z^d\in\mathcal L_d$ we can associate $M:=A^tA\in\mathcal{P}_d$ and the quadratic form $M[x]:=x^tMx$ in $d$ variables. Let $\mathcal L_d^\circ$ be the space of all $d$-dimensional lattices of density $1$: $\mathcal L_d^\circ:=\{A\mathbb Z^d:\ A\in SL(d)\}$, and $\mathcal{P}_d^\circ$ the cone of positive definite matrices of determinant $1$ in $d$ variables. Given a Bravais lattice $L\in \mathcal{L}_d$, its dual lattice $L^*$ is defined by $L^*:=\{x\in \R^2 : \forall p\in L, x\cdot p\in \Z \}$.

%Furthermore, given a Bravais $L\in \mathcal{L}_d$, its covolume $V$ is the volume of its primitive cell, i.e. $L=V^{\frac{1}{d}}L_1$ for some $L_1\in \mathcal{L}_d^\circ$.

\medskip

For $L=A\mathbb Z^d\in \mathcal L_d$ or $L\in \mathcal L_d^\circ$, and $f:(0,+\infty)\to \R$, whose associated matrix and quadratic form are given by $M=A^tA\in\mathcal P_d$ or $\mathcal P_d^\circ$, we define
\begin{equation}\label{energy}
E_f[L]:=\sum_{p\in L\setminus\{0\}} f(|p|^2)=\sum_{x\in\mathbb Z^d\setminus\{0\}}f(M[x]).
\end{equation}
In order for the sum $E_f[L]$ to be equal to an absolutely convergent sum for $L\in\mathcal L^d$, we require that if $F:\mathbb R^d\setminus\{0\}\to(-\infty,+\infty]$ is given by $F(x):=f(|x|^2)$, then $F$ is integrable outside any neighborhood of the origin. For this, we consider the following class of $f$:
\begin{definition}[admissible $f$]\label{def:admissible}
Let $k\in\mathbb N$. A function $\varphi_k:\R^d\to\R$ which is constant on each one of the cubes $\frac1k[0,1)^d+\frac1k \vec a$, with $\vec a\in \mathbb Z^d$, will be called a \emph{$k$-coarse} function. 

\medskip

A function $f:(0,+\infty)\to\R$ is called \emph{admissible in dimension $d$} if for each $k\in\mathbb N$ there exists $k$-coarse functions $\varphi_k^+,\varphi_k^-\in L^1(\mathbb R^d)$ such that $\varphi_k^-(x)\le f(|x|^2)\le \varphi_k^+(x)$ for all $x\in\mathbb R^d\setminus[-1/k,1/k]^d$. 

\medskip

A function $f:(0,+\infty)\to\R$ is called \emph{weakly admissible in dimension $d$} if 
\begin{equation}\label{admiss_eq}
 r^{d-1}f(r^2)\in\bigcap_{\epsilon>0}L^1([\epsilon,+\infty)).
\end{equation}

\end{definition}
\begin{rmk}Note that admissibility implies weak admissibility. Moreover, for $f:(0,+\infty)\to\R$ which is monotone (in particular for $f$ that has positive inverse Laplace transform) or which satisfies $C^1$-bounds away from the origin, the other implication holds: if such $f$ is weakly admissible in dimension $d$ then $f$ is admissible in dimension $d$. The weak admissibility condition \eqref{admiss_eq} in general does not guarantee that $E_f[L]$ is absolutely summable for $L\in\mathcal L^d$, as such $f$ could blow up at $r^2$ corresponding to the distances in lattice $L$, and this is why the more complicated condition in terms of $k$-coarse functions seems justified. On the other hand the simpler condition \eqref{admiss_eq} will suffice for the very regular class of potentials which are treated in this paper.
\end{rmk}

The above functional $E_f$ is invariant under rotations $R\in O(d)$. Indeed, spaces $\mathcal L_d, \mathcal L_d^\circ$ (respectively, $\mathcal P_d,\mathcal P_d^\circ$) have natural actions by rotations, defined as follows. For $R\in O(d)$ and $A\mathbb Z\in \mathcal L_d$ (respectively, $M\in\mathcal P_d$), we define $R\cdot L:=RA\mathbb Z$ (respectively, $R\cdot M:=R^tMR=R^{-1}MR$, which is the induced action under the identification $M=A^tA$, because then $R\cdot M=(RA)^t(RA)$). As $E_f$ only depends on $M[x]$, which in turn is invariant under our $O(d)$-action, we find that $E_f[L]=E_f[R\cdot L]$ for all $R\in O(d)$, as claimed. We denote by $D_{\mathcal L_d},D_{\mathcal L_d^\circ},D_{\mathcal P_d}, D_{\mathcal P_d^\circ}$ fundamental domains for the above actions. Furthermore:

\begin{itemize}
\item If $L$ is a lattice, let $r_0(L):=0<r_1(L)<r_2(L)<\cdots<r_n(L)<\cdots$ be an enumeration of the set of distances $\{|v|:v\in L\}$. In this case we call $L^{(j)}:=\{v\in L:\ |v|=r_j\}$ the $j$-th shell of $L$. We also denote $L^{(j,k)}:=L^{(j)}\cup L^{(k)}$ and $L^{(\ge j)}:=\bigcup_{k\ge j}L^{(k)}$.
\item Let $\mathcal L_d^1$ be the space of all $d$-dimensional lattices whose shortest nonzero vector has lenght $1$. In terms of the above notation, $\mathcal L_d^1:=\{L\in \mathcal L_d:\ r_1=1\}$.
%\item We denote by $\mathsf{A}_2\in\mathcal L_2^1$ the triangular lattice of shortest nonzeo vector with length $1$, i.e. \lb{$\mathsf{A}_2=\sqrt{\frac{\sqrt{3}}{2}}\Lambda_1$.}
\end{itemize}

\begin{definition}[one-well potentials]\label{defn:onewell}We call $f:(0,+\infty)\to \mathbb R\cup\{+\infty\}$ a \emph{one-well potential} if there exists $a>0$ such that $f$ is nonincreasing on $(0,a)$ and nondecreasing on $(a,+\infty)$.
\end{definition}

\begin{definition}[shape of a lattice]\label{defn:shape}
Define an equivalence relation $\sim$ on $\mathcal L_d$ by stating that for $L,L^\prime\in\mathcal L_d$ there holds $L\sim L^\prime$ if there exists $\lambda>0$ and $R\in O(d)$ such that $L=\lambda R\cdot L^\prime$. The intersection of the equivalence class $[L]$ of $L\in\mathcal L_d$ under $\sim$, with the fundamental domain $D_{\mathcal L_d^\circ}$ of $\mathcal L_d^\circ$ under  $O(d)$ action is called the \emph{shape} of $L$.
\end{definition}

\section{Survey of known results in dimension $d=1$}\label{sec:1d}
\begin{rmk}
Note that in dimension $1$ it is common usage to define
\begin{equation}\label{energy1d}
E^\mathrm{1D}_g[L]:=\sum_{p\in L\setminus\{0\}} g(|p|) \quad \text{and}\quad E^\mathrm{1D}_g[L_N]:=\frac{1}{N}\sum_{j=1}^N\sum_{i=-\infty\atop i\neq 0}^{+\infty} g(x_{i+j}-x_i),
\end{equation}
where $L\in\mathcal L_1$ is a lattice and $L_N$ is an $N$-periodic configuration of density $\rho$, i.e. it satisfies $x_{i+N}-x_i=\rho^{-1}$ for some $\rho>0$. We usually suppose that $g(-x)=g(x)$ for all $x\neq 0$, and we note that then we just sum $g(|p|)$ (rather than $g(|p|^2)$, as done here in \eqref{energy}).
\end{rmk}
\subsection{The asymptotics in periodic fixed-density situations}
Recall that, corresponding to the setting in which the shape of $g$ does not automatically determine one scale for minimizers, as explained in Section \ref{sec:intro_gauss}, for $d\ge 2$ we resorted to writing $g(r^2)$ as a superposition of Gaussians, i.e. to looking at the inverse Laplace transform of $g$, and discussing the ensuing coefficients. This setting turns out to be too restrictive in $d=1$. For example we have the following representative result, in a periodic case, and we refer to \cite{BHS} for several further theorems:
\begin{theorem*}[{equivalent to \cite[Prop. 1(A)]{BHS} and \cite[Thm I]{VN1}}]
Let $g:\mathbb R\to \mathbb R\cup \{+\infty\}$ be an even, lower semicontinuous function, invariant under translations by $N\mathbb Z$, so that the values of $g$ on $[0,N/2]$ completely determine $g$. If $g$ is convex decreasing on $[0,N/2]$, then the energy $E^\mathrm{1D}_g[L_N]$ attains a global minimum on the periodic configuration $\mathbb Z$.
\end{theorem*}
In the same work, it is proven that if $g$ as above is \emph{concave and decreasing} on $[0,N/2]$, then the minimum is very degenerate and the points concentrate on a translated copy of $N/2\mathbb Z$, and it is a configuration with multiplicity $\lfloor N/2\rfloor$. Originally this last phenomenon was observed in \cite[Section 5]{VN1}. 

\medskip

A related result appears in \cite{georgakoul}, where an equivalent proof of precisely the same statement is given depending on the \emph{force} between points, giving a satisfactory answer to the question: ``which configurations on the real line are in equilibrium with respect to repulsive $2$-point forces which are strictly decreasing with respect to the distance'', the answer being that only a ``crystalline'' periodic configuration can be stable.

\medskip

The following questions are at the moment still open:
\begin{question}\label{q:necsuf1d}
What is a necessary and sufficient condition on $g$ under which for any $N\in\mathbb{N}^*$ all minimizing $1$-dimensional $N$-periodic configurations of density $\rho$ are up to translation equal to the lattice $\rho^{-1}\mathbb Z$? 
\end{question}
\begin{question}\label{q:necsuf_force1d}
What is a necessary and sufficient condition on a repulsive force $G$ under which for any $N\in\mathbb{N}^*$ a $1$-dimensional $N$-periodic configuration of density $\rho$ in which forces balance at each point, are up to translation equal to the lattice $\rho^{-1}\mathbb Z$? 
\end{question}

\subsection{One-well potentials $g$ without constraints on the density}

\begin{itemize}
\item In \cite{vn1d1} it was shown that there exist nonconvex potentials $g$ which have $\rho^{-1}\mathbb Z$ as unique minimizer for its average energy per particle among $N$-periodic configurations  of density $\rho$ for any $N$ and any $\rho$.
\item In \cite[p. 284]{vn1d1}, for $g(x)=(1+|x|^4)^{-1}$, it has been proved that the average energy of the regular $2$-periodic configuration $1/2\Z$ is larger than the $2$-periodic most degenerate case, in which two points share each position $n\in \Z$.
\item In \cite{vn1d2} it was shown that in the class of $f$ such that $g(x)=g(-x)$, such that $g''$ exists and is well-behaved at infinity, and such that the Fourier transform $\hat g$ exists, a \textbf{necessary condition} for $\mathbb Z$ to satisfy the optimality condition among $N$-periodic configurations for any $N$ at high enough density $\rho\ge \rho_0$, is that $\hat g\ge 0$. This results has been generalized to two-component systems with three kind of interacting potentials in \cite{BetKnupfNolte}.
\item In \cite{gardnerradin} it was proved that for the classical Lennard-Jones potential $g(r)=r^{-12}-r^{-6}$ the unique minimizer amongst all configurations in dimension $d=1$ is a periodic crystal. This was later extended in \cite{radinclassground} to more general potentials of the form $g(r)=g_1(r)-g_2(r)$ with $g_1,g_2$ convex and with good decay properties, based on Sinai's theorem on thermodynamic limits and on the previous work \cite{VN1}, and necessary conditions for crystallization were given. These conditions seem relatively cumbersome, and as far as we could check, they were not further improved in later works.
\end{itemize}
The following question seems to be still open:
\begin{question}
In dimension $d=1$, what is the largest class of potentials $g$ which have the property that the minimizers of $N$-point energies asymptotically as $N\to +\infty$ approximate (up to translation and rotation) a periodic configuration?
\end{question}

\section{Optimality and non-optimality at all scales}\label{sect-allscales}
We recall the following well-known result in dimension $d=2$, which follows from the Montgomery \cite[Thm 1]{Mont} and Bernstein-Hausdorff-Widder \cite{Bernstein} theorems (see the statements of these theorems in the Introduction):
\begin{proposition}[{see e.g. \cite[p. 169]{CohnKumar} or \cite[Prop 3.1]{BetTheta15}}]\label{prop:completemonot}
If $f$ is completely monotone, then for any $\lambda>0$, the triangular lattice $\Lambda_1$ is the unique minimizer of $L\mapsto E_f[\lambda L]$ in $D_{\mathcal L_2^\circ}$.
\end{proposition}
The above sufficient condition is not necessary, due to our new counterexample of Section~\ref{not_cm}. On the other hand, too simple criteria do not furnish sufficient conditions strong enough to replace complete monotonicity. Indeed, by methods related to Montgomery's approach, in \cite[Prop 3.4]{BetTheta15}, the first author also proved that the following positive, decreasing and convex function
\begin{equation}\label{Counterexample}
V(r)=\frac{14}{r^2}-\frac{40}{r^3}+\frac{35}{r^4}
\end{equation}
is such that the triangular lattice is not a minimizer of $E_V$ in $\lambda \mathcal L_2^\circ$ where $\lambda_1<\lambda<\lambda_2$ for values $\lambda_1\approx 1.522$ and $\lambda_2\approx 1.939$.

\medskip

In the two next subsections, we are prove Theorems \ref{MainTh1} and \ref{MainTh2} about the minimality of our special lattices at fixed scale.
\subsection{Non-optimality at high or low density for a subclass of functions}
In this part, we consider $f$ admissible in $\mathbb R^d$ and such that %$d\mu_f(x)=\rho_f(x)dx$ and
\begin{equation}\label{def-flaplace}
f(r)=\int_0^{+\infty} e^{-rt} d\mu_f(t)
%\rho_f(t)dt
:=\mathcal L[\mu_f](r),
\end{equation}
where $\mu_f:=\mathcal{L}^{-1}[f]$ is the inverse Laplace transform of $f$, which we assume to be a well-defined Radon measure.

\medskip

We start by recalling Jacobi's transformation formula for the lattice theta function defined by
\begin{equation}\label{def-theta}
\theta_L(\alpha)=\sum_{p\in L} e^{-\pi \alpha |p|^2}.
\end{equation}
which is in fact a simple application of Poisson Summation Formula. A proof of this identity can be found for instance in \cite{Bochnertheta}.
\begin{lemma}[Jacobi's transformation formula]\label{lem:mont}
For any $d\geq 1$, any $\alpha>0$ and any Bravais lattice $L\in  \mathcal{L}_d^\circ$,
\begin{equation}\label{identity-thetaA}
\theta_L(\alpha)=\alpha^{-\frac{d}{2}}\theta_{L^*}\left(\frac{1}{\alpha}\right).
\end{equation}
\end{lemma}
From \eqref{def-flaplace} and \eqref{identity-thetaA}, it is possible to write $E_f[L]$ in terms of $\theta_L$ as follows:
\begin{proposition}\label{prop:tri2}
For any $\lambda>0$, any admissible function $f$ having the representation \eqref{def-flaplace} with absolutely continuous $d\mu_f(t)=\rho_f(t)dt$ and any Bravais lattice $L\in \mathcal{L}_d^\circ$, we have
\begin{align}
E_f[\lambda L]&=\frac{\pi}{\lambda^2}\int_0^{+\infty} \left(\theta_{L}(u)-1  \right)\rho_f\left( \frac{\pi u}{\lambda^2} \right)du \label{identity-Ef1}\\
&=\frac{\pi}{\lambda^2}\int_0^{+\infty} \left(u^{\frac{d}{2}}\theta_{L^*}(u)-1  \right)\rho_f\left( \frac{\pi}{\lambda^2u} \right)u^{-2}du.\label{identity-Ef2}
\end{align}
\end{proposition}
\begin{proof}
The first equality is clear by definition of $f$, by the change of variables $u=\frac{\lambda^2 t}{\pi}$,
\begin{align*}
E_f[\lambda L]=\int_0^{+\infty}\left( \theta_{\lambda L}\left( \frac{t}{\pi} \right)-1 \right)\rho_f(t)dt &=\int_0^{+\infty}\left( \theta_{L}\left( \frac{ \lambda^2 t}{\pi} \right)-1 \right)\rho_f(t)dt \\
&=\frac{\pi}{\lambda^2}\int_0^{+\infty} \left(\theta_{L}(u)-1  \right)\rho_f\left( \frac{\pi u}{\lambda^2} \right)du.
\end{align*}
The second equality is proved using \eqref{identity-thetaA} for $\alpha=\frac{\lambda^2 t}{\pi}$ and by change of variables $u=\frac{\pi}{t \lambda^2}$:
\begin{align*}
E_f[\lambda L] &=\int_0^{+\infty}\left( \theta_{L}\left( \frac{ \lambda^2 t}{\pi} \right)-1 \right)\rho_f(t)dt \\
&=\int_0^{+\infty} \left( \left(  \frac{\pi}{t \lambda^2}\right)^{\frac{d}{2}} \theta_{L^*}\left(\frac{\pi}{t \lambda^2}  \right)-1 \right)\rho_f(t)dt\\
&=\frac{\pi}{\lambda^2}\int_0^{+\infty} \left(u^{\frac{d}{2}}\theta_{L^*}(u)-1  \right)\rho_f\left( \frac{\pi}{\lambda^2 u} \right)u^{-2}du.
\end{align*}
\end{proof}
We next assume that the dimension $d$ is such that there exists only one lattice $L_m\in \mathcal L_d^\circ$ which is the unique minimizer of $L\mapsto \theta_L(\alpha)$ on $D_{\mathcal L_d^\circ}$ for any fixed $\alpha>0$. This is known to be the case for $d=2$, where $L_m=\Lambda_1$ by Montgomery Theorem \cite[Thm 1]{Mont}. We now prove our first result about the non-optimality of $L_m$ for some admissible potentials $f$ that have the representation \eqref{def-flaplace}, when its inverse Laplace transform is negative in the neighbourhood of $0$ or $+\infty$.

\begin{proposition}\label{prop-asymptnonopt}
Assume that there exists a lattice $L_m\in \mathcal{L}_d^\circ$ which is the unique minimizer of $L\mapsto \theta_L(\alpha)$ on $D_{\mathcal L_d^\circ}$ for any fixed $\alpha>0$. Let $f$ be an admissible potential having the representation \eqref{def-flaplace}, then the following holds for any $r_0>0$:
\begin{enumerate}
\item If $\mu_f<0$ on $(0,r_0)$, then for any Bravais lattice $L\in \mathcal{L}_d^\circ\backslash\{L_m\}$, there exists $\lambda_0$ such that for any $\lambda>\lambda_0$, $E_f[\lambda L]<E_f[\lambda L_m]$;
\item  If $\mu_f<0$ on $(r_0,+\infty)$, then for any Bravais lattice $L\in \mathcal{L}_d^\circ\backslash\{L_m\}$, there exists $\lambda_1$ such that for any $0<\lambda<\lambda_1$, $E_f[\lambda L]<E_f[\lambda L_m]$.
\end{enumerate}
In particular, these results hold in dimension $d=2$ for the triangular lattice $L_m=\Lambda_1$.
\end{proposition}
\begin{rmk}\label{rmk:824}
This result would hold in dimensions $8$ and $24$ for $\mathsf{E}_8$ or the Leech lattice $\Lambda_{24}$, once the universality of these lattice, i.e. their minimality for the theta function among periodic configurations of fixed unit density and for any $\alpha>0$, conjectured in \cite[Conjecture 9.4]{CohnKumar}, will be proved. The same result in dimension $d=4$ could also be proved for $\mathsf{D}_4$, according to its local minimality for the lattice theta function proved in \cite{Coulangeon:2010uq}.
\end{rmk}
\begin{proof}
Assume first that $\mu_f$ is absolutely continuous with respect to the Lebesgue measure and $\mu_f(t)=\rho_f(t)dt$. If the hypothesis of point (1) holds, for any Bravais lattice $L\in \mathcal{L}_d^\circ$, we write, using \eqref{identity-Ef1},
\begin{equation}
E_f[\lambda L]-E_f[\lambda L_m]=\frac{\pi}{\lambda^2}\int_0^{+\infty} \left(\theta_{L}(u)-\theta_{L_m}(u)  \right)\rho_f\left( \frac{\pi u}{\lambda^2} \right)du.
\end{equation}
By assumption, $\theta_{L}(u)-\theta_{L_m}(u)>0$ for all $u>0$ and $\rho_f\left( \frac{\pi u}{\lambda^2} \right)<0$ if $u<\frac{r_0 \lambda^2}{\pi}$. By the exponential decay of $u\mapsto\theta_{L}(u)$ for any fixed $L\in \mathcal{L}_d^\circ$, we obtain that for any Bravais lattice $L\in \mathcal{L}_d^\circ$, there exists $\lambda_0$ such that for any $\lambda>\lambda_0$, 
$$
\int_0^{+\infty} \left(\theta_{L}(u)-\theta_{L_m}(u)  \right)\rho_f\left( \frac{\pi u}{\lambda^2} \right)du<0,
$$
 and the first part of the proposition is proved.

\medskip

For the second case, by \eqref{identity-Ef2}, we get
\begin{equation}
E_f[\lambda L]-E_f[\lambda L_m]=\frac{\pi}{\lambda^2}\int_0^{+\infty} \left(\theta_{L^*}(u)-\theta_{L_m^*}(u)  \right)\rho_f\left( \frac{\pi}{\lambda^2 u} \right)u^{\frac{d}{2}-2}du.
\end{equation}
It is clear from \eqref{identity-thetaA} and the fact that $L_m$ is the unique minimizer of the lattice theta function for all $\alpha>0$ that we necessarily have $L_m^*=L_m$. By assumption, $\theta_{L^*}(u)-\theta_{L_m^*}(u)>0$ for all $u>0$ and $\rho_f\left( \frac{\pi }{\lambda^2 u} \right)<0$ if $u<\frac{\pi}{\lambda^2 r_0}$. As in the previous case, by the exponential decay of the lattice theta function, there exists $\lambda_1$ such that for any $0<\lambda<\lambda_1$, 
 $$
\int_0^{+\infty} \left(\theta_{L^*}(u)-\theta_{L_m^*}(u)  \right)\rho_f\left( \frac{\pi}{\lambda^2 u} \right)u^{\frac{d}{2}-2}du<0
 $$
 and the second case of the proposition is proved.
 
 \medskip
 
 If $\mu_f$ is a Radon measure but not absolutely continuous with respect to the Lebesgue measure, then we can repeat the proof of Proposition \ref{prop:tri2} and prove 
\[  
E_f[\lambda L]=\int_0^{+\infty} \left(\theta_{L}(u)-1  \right)\bar\mu_f(u)=\int_0^{+\infty} \left(u^{\frac{d}{2}}\theta_{L^*}(u)-1  \right)\tilde \mu_f(u),
\]
where $\bar \mu_f=(g_\lambda)_\#\mu_f$ with $g_\lambda(t)=\frac{\lambda^2t}{\pi}$ and $\tilde \mu_f=(h_\lambda)_\#\mu_f$ with $h_\lambda(t)=\frac{\pi}{t\lambda^2}$, and the proof continues as above.
\end{proof}

\subsection{About lattices that are optimal for any density}
The next result shows that, if the inverse Laplace transform of $f$ has a sign in a neighbourhood of the origin or $+\infty$ (which is the case for all the classical example or functions constructed with inverse power laws, exponentials, Yukawa potentials, Gaussians...) and if $L_m$ is the unique minimizer of $L\mapsto \theta_L(\alpha)$ in $D_{\mathcal{L}_d^\circ}$ for all $\alpha>0$, then the only lattice that could be a minimizer for all the densities is $L_m$. As recalled in the previous section (see Proposition \ref{prop-asymptnonopt} and Remark \ref{rmk:824}), this is the case in dimension $d=2$ for $L_m=\Lambda_1$ is the triangular lattice and is conjectured to extend in dimensions $d\in\{4,8,24\}$.
\begin{proposition}\label{prop-anyscale2}
Let $d\geq 1$ and assume that $L_m$ is the unique minimizer of $L\mapsto \theta_L(\alpha)$ in $D_{\mathcal{L}_d^\circ}$ for any $\alpha>0$. Let $f$ be an admissible potential with representation \eqref{def-flaplace} such that $\mu_f>0$ on the interval $(0,r_0)$ or on on the interval $(r_0,+\infty)$. If $\Lambda$ is a minimizer of $L\mapsto E_f[\lambda L]$ for any $\lambda >0$ on $D_{\mathcal{L}_d^\circ}$, then $\Lambda=L_m$.

\medskip

Furthermore, if $\mu(r)<0$ on $(0,r_0)$ or on $(r_0,+\infty)$, then the minimizer of $L\mapsto E_f[\lambda L]$ cannot be the same for all $\lambda>0$.

\medskip

In particular, this result holds in dimension $d=2$ for the triangular lattice $L_m=\Lambda_1$.
\end{proposition}
\begin{proof}
 We perform the proof under the assumption that $\mu_f$ is absolutely continuous with respect to the Lebesgue measure, i.e. $d\mu_f(t)=\rho_f(t)dt$, and the general case follows by the same adaptation as in the proof of Proposition \ref{prop-asymptnonopt}.

\medskip
 
We prove the first point. Let $\Lambda$ be a minimizer of $L\mapsto E_f[\lambda L]$ on $D_{\mathcal{L}_d^\circ}$ for any $\lambda>0$. We assume that $\Lambda\neq L_m$. Therefore, by strict minimality of $L_m$ and continuity of $L\mapsto \theta_L(u)$ for any fixed $u>0$, there exists $L_0\in D_{\mathcal{L}_d^\circ}$ such that $\theta_{L_0}(u)-\theta_{\Lambda}(u)<0$ for all $u>0$. We now use exactly the same approach as in Proposition \ref{prop-asymptnonopt}, using \eqref{identity-Ef1} and \eqref{identity-Ef2}. In the first case, we write by assumption
\begin{equation}\label{assumpt}
\forall V>0, \quad 0<E_f[\lambda L_0]-E_f[\lambda\Lambda]=\frac{\pi}{\lambda^2}\int_0^{+\infty} \left(\theta_{L_0}(u)-\theta_{\Lambda}(u)  \right)\rho_f\left( \frac{\pi u}{\lambda^2} \right)du.
\end{equation}
Since $\rho_f\left( \frac{\pi u}{\lambda^2} \right)<0$ if $u<\frac{r_0 \lambda^2}{\pi}$, by the exponential decay of $u\mapsto \theta_L(u)$ for any Bravais lattice $L\in \mathcal{L}_d^\circ$, there exists $\lambda_0$ such that for any $\lambda>\lambda_0$,
\begin{equation}
\frac{\lambda^2}{\pi}\left(E_f[\lambda L_0]-E_f[\lambda \Lambda]\right)=\int_0^{+\infty} \left(\theta_{L_0}(u)-\theta_{\Lambda}(u)  \right)\rho_f\left( \frac{\pi u}{\lambda^2} \right)du<0,
\end{equation}
that contradicts \eqref{assumpt}. The second case is proved similarly using the second equality in \eqref{identity-Ef2}, as in the proof of Proposition \ref{prop-asymptnonopt}.

\medskip

For the second point, i.e. $\mu$ is negative in the neighbourhood of $0$ or $+\infty$, we use the same arguments. By Proposition \ref{prop-asymptnonopt}, $L_m$ cannot be a minimizer of $L\mapsto E_f[\lambda L]$ for any $\lambda>0$. Let us assume that there exists another Bravais lattice $\Lambda\in D_{\mathcal{L}_d^\circ}$ that is a minimizer of the energy for all $\lambda>0$. Therefore, there exists $L_0$ such that, for any $u>0$, $\theta_{L_0}(u)-\theta_\Lambda(u)>0$. If $\mu(r)<0$ on $(0,r_0)$ the same argument as before shows that there exists $\lambda_0$ such that for any $\lambda>\lambda_0$, $E_f[\lambda L_0]-E_f[\lambda \Lambda]<0$ and that contradicts our assumption. The second case is proved similarly as explained before and shows the existence of $\lambda_1$ such that for any $0<\lambda<\lambda_1$, $E_f[\lambda L_0]-E_f[\lambda \Lambda]<0$, that again contradicts our assumption.
\end{proof}
\section{One-well potentials and optimality of lattices}\label{sectLJ}
\subsection{Potentials with one well}
This Section is devoted to the proof of Theorem \ref{MainTh3} and \ref{MainTh4}. We also give several numerical evidence for the minimality of certain lattices for the Lennard-Jones type energy.
\subsection{Lennard-Jones type potentials}
In this section, we are discussing the minimization problem for the Lennard-Jones type potentials defined by
\begin{equation}\label{lj2}
f(r)=\frac{a_2}{r^{x_2/2}}-\frac{a_1}{r^{x_1/2}}, \quad x_2>x_1>d\quad (a_1,a_2)\in (0,+\infty)^2,
\end{equation}
and, for any Bravais lattice $L\in \mathcal L_d$,
\begin{equation}
E_f[L]=a_2 \zeta_L(x_2)-a_1 \zeta_L(x_1).
\end{equation}
We first prove that the shape of a global minimizer $L_0$  (which is defined as a canonical choice of a lattice equivalent to $L_0$ under rotation and dilation, see Definition~\ref{defn:shape}) of $E_f$ for fixed parameters $(a_1,a_2,x_1,x_2)$ does not depend on $(a_1,a_2)$.
\begin{proposition}\label{prop-LJglobal}
Let $d<x_1<x_2$ and $(a_1,a_2)\in (0,+\infty)^2$, and $L_0=\lambda_1 L_1=\lambda_2 L_2$, where the normalizations of $L_1, L_2$ are chosen so that 
\begin{equation}\label{normaliz_L1}
% \min_{v\in L_1\setminus\{0\}}|v|=1, \quad |L_2|=1\ .
L_1\in\mathcal L_d^\circ,\quad L_2\in\mathcal L_d^1.
\end{equation}
Then for $f$ a Lennard-Jones type potential as in \eqref{lj2}, the following hold:
\begin{enumerate}\item $L_0$ is a global minimizer of $E_f$ on $D_{\mathcal L_d}$ if and only if $L_1$ is a minimizer on $D_{\mathcal L_d^\circ}$ of the energy
\begin{equation}\label{def-Etilde}
\tilde{E}_f[L]:=\frac{\zeta_L(x_2)^{x_1}}{\zeta_L(x_1)^{x_2}}.
\end{equation}
\item If the minimum of $f(r)$ is achieved at $r_f=r_f(x_1,x_2,a_1,a_2)>0$, and $\lambda^{L_2}_0=\lambda^{L_2}_0(x_1,x_2,a_1,a_2)>0$ is the factor such that $\lambda^{L_2}_0L_2$ realizes the minimum energy among lattices of the same shape as $L_2$, i.e. $\min_{\lambda>0} E_f[\lambda L_2]=E_f[\lambda^{L_2}_0 L_2]$, then the following limit exists and satisfies
\begin{equation}\label{limit_r0}
\lim_{\substack{x_1,x_2\to+\infty\\x_1<x_2, r_f=r_0}}\lambda^{L_2}_0 (x_1,x_2,a_1,a_2)=\sqrt{r_0}.
\end{equation}
\end{enumerate}
\end{proposition}
\begin{proof}
For any Bravais lattice $L\in\mathcal L_d^\circ$ and any $\lambda>0$, we have
\begin{equation}
E_f[\lambda L]=\frac{a_2}{\lambda^{x_2}}\zeta_L(x_2)-\frac{a_1}{\lambda^{x_1}}\zeta_L(x_1)\ .
\end{equation}
Therefore, with the notation for $\lambda^L_0$ as in above point (2), we get
\begin{equation}\label{value_AL}
\partial_\lambda E_f[\lambda L]\geq 0 \iff \lambda\geq \lambda^L_0,\quad \mbox{and}\quad \lambda^L_0=\left( \frac{a_2x_2\zeta_L(x_2)}{a_1x_1\zeta_L(x_1)} \right)^{\frac{1}{(x_2-x_1)}}\ .
\end{equation}
It follows that, for any $L \in\mathcal L_d^\circ$, 
\begin{equation}\label{min_A}
\min_{\lambda>0} E_f[\lambda L]=E_f[\lambda^L_0 L]=\frac{a_1^{\frac{x_2}{x_2-x_1}}\zeta_L(x_1)^{\frac{x_2}{x_2-x_1}}}{a_2^{\frac{x_1}{x_2-x_1}}\zeta_L(x_2)^{\frac{x_1}{x_2-x_1}}}\left( \left(\frac{x_1}{x_2}\right)^{\frac{x_2}{x_2-x_1}} -\left(\frac{x_1}{x_2}\right)^{\frac{x_1}{x_2-x_1}} \right)<0.
\end{equation}
Therefore, we have, for any Bravais lattices $L,L^\prime\in\mathcal L_d^\circ$,
\begin{equation}
\min_{\lambda>0} E_f[\lambda L]\leq \min_{\lambda>0} E_f[\lambda L^\prime]\iff  \tilde{E}_f[L]\leq \tilde{E}_f[L^\prime]\ ,
\end{equation}
which proves point (1) of the Proposition. Then by computing the minimum of $f$ we may find
\begin{equation}\label{rel_rf_AL}
r_f=\left(\frac{a_2x_2}{a_1x_1}\right)^{\frac2{(x_2-x_1)}}\quad\mbox{and}\quad \lambda^L_0\stackrel{\eqref{value_AL}}{=}\sqrt{r_f}\left(\frac{\zeta_L(x_2)}{\zeta_L(x_1)}\right)^{\frac1{(x_2-x_1)}}\ .
\end{equation}
Now in order to prove point (2) we note that due to \eqref{normaliz_L1} we have $\zeta_{L_2}(x)>1$ for all $x>0$ and for each $\epsilon>0$ there exists a finite bound $C_\epsilon>0$ such that $\zeta_{L_2}(x)\le C_\epsilon$ for all $x>d+\epsilon$. Moreover, we have 
\begin{equation}\label{limit_L1}
\lim_{x\to+\infty}\zeta_{L_2}(x)=\#L_2^{(1)}\ .
\end{equation}
From \eqref{limit_L1} and \eqref{rel_rf_AL} we find 
\begin{equation*}
\lim_{\substack{x_1,x_2\to+\infty\\x_1<x_2, r_f=r_0}}\lambda^{L_2}_0(x_1,x_2,a_1,a_2)= \sqrt{r_0}\lim_{\substack{x_1,x_2\to+\infty\\x_1<x_2, r_f=r_0}}\left(\frac{\zeta_{L_2}(x_2)}{\zeta_{L_2}(x_1)}\right)^{\frac1{x_2-x_1}}= \sqrt{r_0}\ ,
\end{equation*}
which proves \eqref{limit_r0} and concludes the proof.
\end{proof}
Next, extending point (1) of Proposition~\ref{prop-LJglobal}, we formulate some simple equivalent condition for $\Lambda$ to minimize $E_f[L]$ among lattices of unit density:
\begin{proposition}\label{prop:equiv}
Let $L_0=\lambda \Lambda$ for $\Lambda\in \mathcal{L}_d^\circ$. If $L\in \mathcal L_d^\circ$ then we define functions $H_L,h_L:(d,+\infty)\to \R$ by
$$
H_L(x):=\frac{1}{x}\log\left(  \frac{\zeta_L(x)}{\zeta_{\Lambda }(x)}\right), \quad
h_L(x):=-\log \zeta_L(x) +x\frac{\partial_x\zeta_L(x)}{\zeta_L(x)}.
$$
The following conditions are equivalent:
\begin{enumerate}
\item For any $x_2>x_1>d$, the lattice $L_0 $ is the unique minimizer of $E_f$ on $D_{\mathcal L_d}$.
\item  For any $x_2>x_1>d$, the lattice $\Lambda $ is the unique minimizer of $\tilde{E}_f$, defined by \eqref{def-Etilde}, on $D_{\mathcal L_2^\circ}$.
\item For any Bravais lattice $L\in \mathcal L_d^\circ\setminus\{ \Lambda \}$, the function $H_L$ is strictly increasing on $(d,+\infty)$.
\item For any $x>d$, $\Lambda $ is the unique minimizer of $h_L(x)$ on $D_{\mathcal{L}_d^\circ}$.
\end{enumerate}
\end{proposition}
\begin{proof}
The equivalence between (1) and (2) is treated in point (1) of Proposition~\ref{prop-LJglobal}. To prove the equivalence between (2) and (3), note that 
\begin{equation}
\tilde{E}_f[\Lambda]\leq \tilde{E}_f[L]\iff \left( \frac{\zeta_L(x_1)}{\zeta_\Lambda(x_1)} \right)^{\frac{1}{x_1}}\leq \left( \frac{\zeta_L(x_2)}{\zeta_\Lambda(x_2)}  \right)^{\frac{1}{x_2}}.
\end{equation}
For proving the equivalence between (3) and (4), note that $H'(x)=\frac{1}{x^2}\left( h_L(x)-h_{\Lambda }(x) \right)$.
\end{proof}
\begin{rmk}
We conjecture that the equivalent statements of Proposition~\ref{prop:equiv} hold true for the triangular lattice $\Lambda=\Lambda_1=\frac{2}{\sqrt 3}\mathsf{A}_2$ in dimension $d=2$ and for $\Lambda\sim \mathsf{D}_3, \sim\mathsf{D}_4, \sim\mathsf{E}_8,\sim\Lambda_{24}$ respectively in dimensions $3,4,8$ and $24$, where the equivalence relation is as in Definition~\ref{defn:shape} and indicates that $\Lambda$ has the shape of the corresponding lattices.

\medskip

About the case $d=2$, (3) of Lemma~\ref{prop:equiv} has been proved in \cite[Thm. 1.2.B.2 and Lem. 6.17]{BetTheta15} on the small interval $I=(2,2\psi^{-1}(\log \pi)-2)$ where $\psi$ is the digamma function and $2\psi^{-1}(\log \pi)-2\approx 5.256$. Figure \ref{H2to10square}, where we have plotted $H_{\Z^2}$, shows the monotonicity of the function and is then a numerical evidence of the optimality of the triangular lattice against the square lattice. The same conclusion follows from Figure \ref{fig:ratio} where we find that for all $x_1<x_2<2x_1$, that we tested (in particular for all half-integer values of $x_1$, $x_2$ in the range $1<x_1<x_2\le 25$), the triangular lattice has an energy $E_f$ lower than the square lattice.  Moreover, numerically we find that the value $ \frac{\min_{A>0}E_f[\sqrt{A}\Lambda_1]}{\min_{A>0}E_f[\sqrt{A}\mathbb Z^2]}$ is increasing in $x_1$, $x_2$.

\medskip

In dimension $d=3$ we also have plotted in Figure \ref{Z3vsFCCBCC} the function $H_{\Z^3}$ where $\Lambda\in \{\mathsf{D}_3,\mathsf{D}_3^*\}$ are the FCC lattice and the BCC lattice of unit density. Both functions seem to be increasing, showing the optimality of the FCC lattice (resp. BCC lattice) against the cubic lattice. The same occurs if $L$ is the BCC lattice and $\Lambda$ is the FCC lattice as we can see in Figure \ref{fig-BCCFCC}. These numerics are consistent with \cite[Conjecture 1.7]{Beterminlocal3d} which states that optimality of $\mathsf{D}_3$ holds in dimension $3$ for any exponents $x_2>x_1>3$.
\end{rmk}
\begin{figure}[!h]
\centering
\includegraphics[width=8cm]{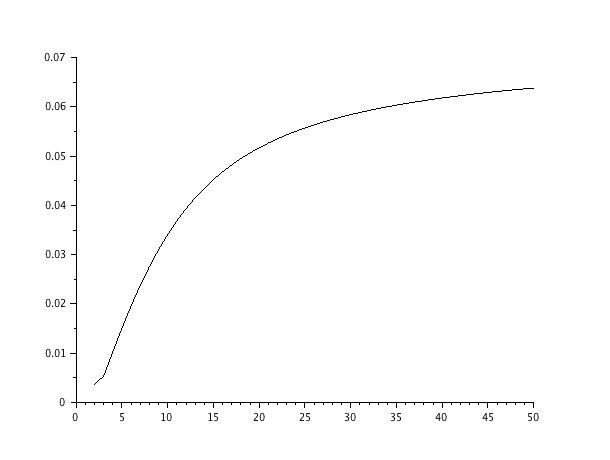}  
\caption{Plot of function $H_{\Z^2}$ on $[2,50]$, for the triangular lattice $\Lambda=\Lambda_1$. The triangular lattice seems to have lower Lennard-Jones type energy than the square lattice for any values of the parameters.}\label{H2to10square}
\end{figure}
\begin{figure}[!h]
\centering
\includegraphics[width=9cm]{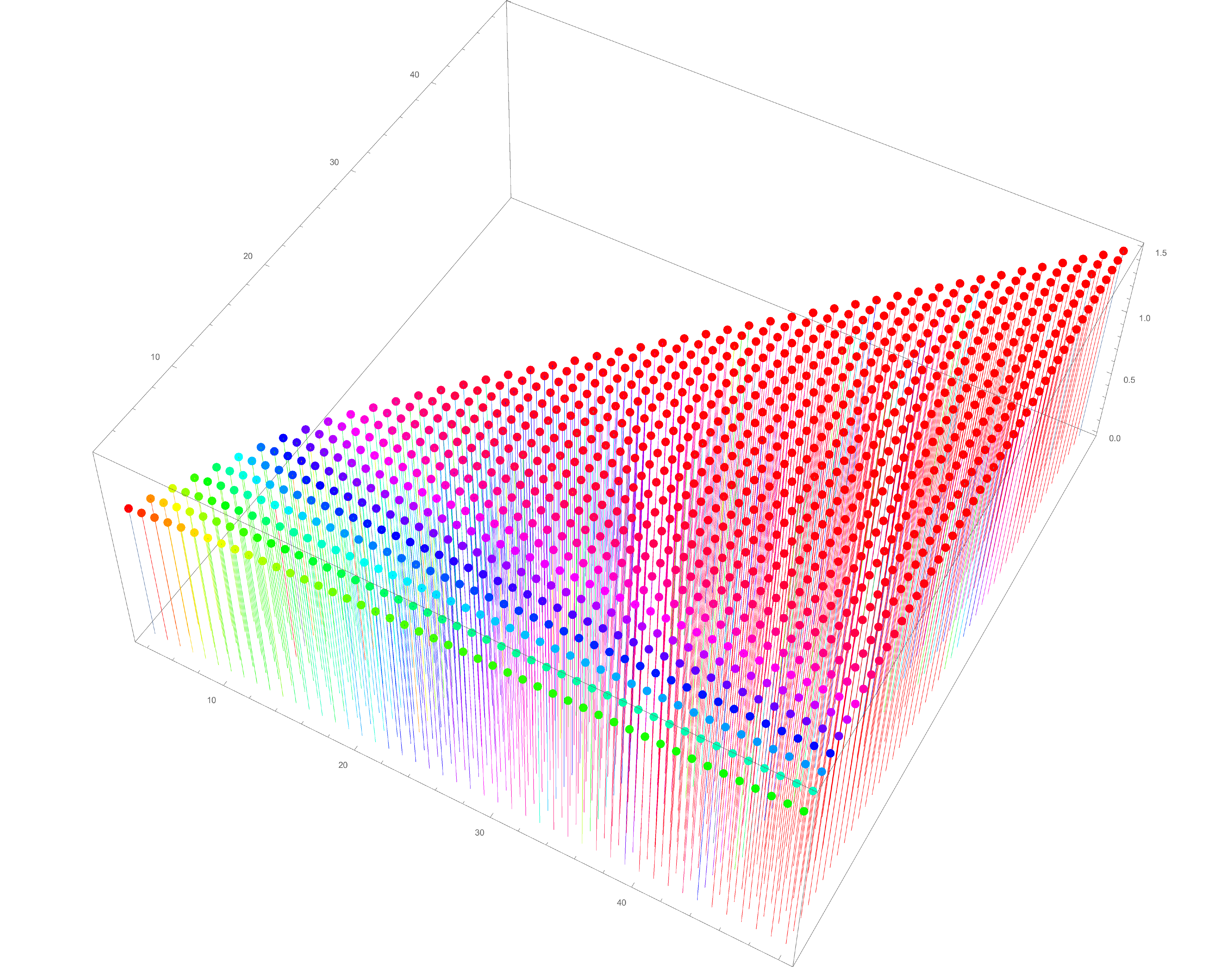} \\
\caption{Plot in the $(x_1,x_2)$-plane of the quantity  $\frac{\min_{\lambda>0}E_f[\lambda\Lambda_1]}{\min_{\lambda>0}E_f[\lambda\mathbb Z^2]}$, for $3\le x_1<x_2\le 50$ and $x_1,x_2\in\mathbb Z$. Note that the smallest value is taken at $x_1=3, x_2=4$ and equals $\approx 1.061$, i.e. is already larger than $1$, and then the value increases in both coordinate directions, so that for $x_1=49$, $x_2=50$ it equals $\approx 1.499$, very close to the asymptotic value $3/2$.}
\label{fig:ratio}
\end{figure}
\begin{figure}[!h]
\centering
\includegraphics[width=6cm]{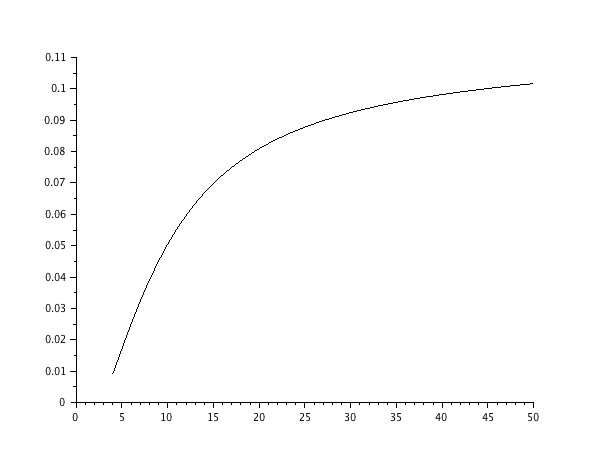}  \quad \includegraphics[width=6cm]{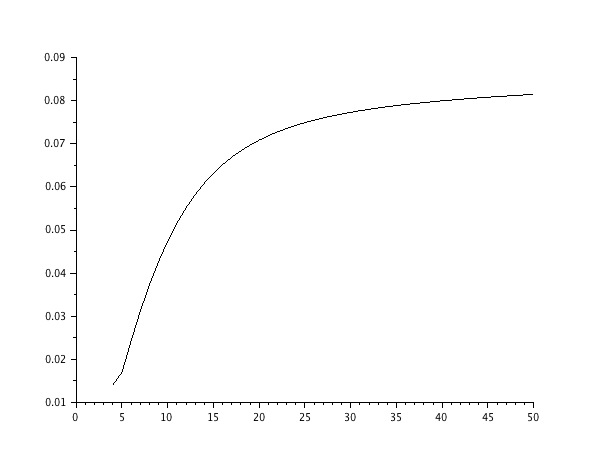}
\caption{Plot of function $H_{\Z^3}$ on $[4,50]$ where $\Lambda=\mathsf{D}_3$ is the FCC lattice (on the left) and $\Lambda=\mathsf{D}_3^*$ the BCC lattice (on the right). The FCC and BCC lattices seem to have lower Lennard-Jones type energy than the cubic lattice for any values of the parameters.}\label{Z3vsFCCBCC}
\end{figure}
\begin{figure}[!h]
\centering
\includegraphics[width=8cm]{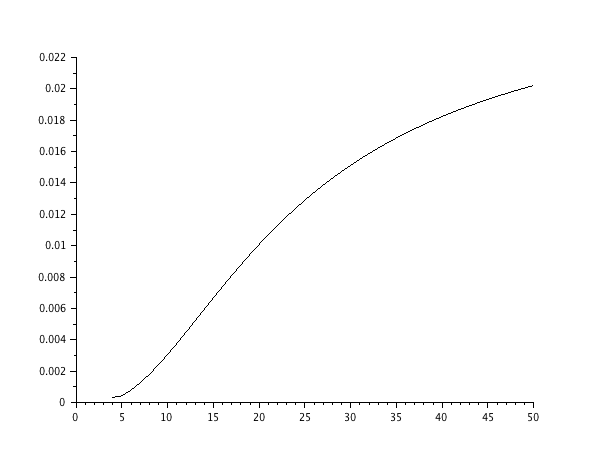}  
\caption{Plot of function $H_{\mathsf{D}_3^*}$ on $[4,50]$ where $\Lambda=\mathsf{D}_3$ is the FCC lattice. The FCC lattice seems to have lower Lennard-Jones type energy than the BCC lattice for any values of the parameters.}\label{fig-BCCFCC}
\end{figure}
\begin{definition}\label{def:kissing} The \emph{kissing number} (also called \emph{coordination number} in crystallography) of a Bravais lattice $L\in\mathcal L_d$, denoted $\tau(L)$, is defined as the number of nearest-neighbors in $L$ of the origin, $\tau(L):=\#L^{(1)}$ in our notation. We also define the \emph{optimal kissing number in dimension $d$} by $k(d):=\max_{L\in\mathcal L_d}\tau(L)$.
\end{definition}
\begin{rmk}\label{rmk:kissing}
It is known (see \cite[Table 1.1]{ConSloanPacking}, \cite{musinkiss} and the references therein) that $k(2)=\tau(\mathsf{A}_2)=6$, $k(3)=\tau(\mathsf{D}_3)=12$, $k(4)=\tau(\mathsf{D}_4)=24$, $k(8)=\tau(\mathsf{E}_8)=240$, $k(24)=\tau(\Lambda_{24})=196560$, while all other cases are not known. The above lattices are known to be unique optimizers among lattices, and in $2,8,24$ dimensions also among general packing configurations, while in dimension $3$ this is not the case, see \cite{kusnerkusner}, and in dimension $4$ it is not known, see \cite{musinkiss}.
\end{rmk}
\begin{proposition}[Asymptotic minimality and kissing number]\label{prop:asymin}
For $L,\Lambda\in\mathcal L_d$ and $f$ a Lennard-Jones type potential as in \eqref{lj2} there holds 
\begin{equation}\label{limit_ratioenergy}
\lim_{x_1\to +\infty}\lim_{x_2\to +\infty}\frac{\displaystyle \min_{\lambda>0}E_f[\lambda\Lambda]}{\displaystyle \min_{\lambda>0}E_f[\lambda L]}=\frac{\tau(\Lambda)}{\tau(L)}\ .
\end{equation}
As the above minimum values are negative, if $\tau(\Lambda)$ uniquely realizes the optimal kissing number $k(d)$ amongst lattices, then there exists $x_0=x_0(d)$ such that for any $x_2>x_1>x_0$, the unique minimizer of $E_f$ in $D_{\mathcal{L}_d}$ has shape $[\Lambda]$. In particular, this holds for the cases $(d,\Lambda)\in\{(2,\mathsf{A}_2), (3,\mathsf{D}_3), (4,\mathsf{D}_4),(8,\mathsf{E}_8),(24,\Lambda_{24})$.
\end{proposition}
\begin{proof}
Due to \eqref{min_A} we have, for any Bravais lattices $L,\Lambda\in\mathcal L_d^1$,
\begin{equation}
\frac{\displaystyle \min_{\lambda>0}E_f[\lambda \Lambda]}{\displaystyle \min_{\lambda>0}E_f[\lambda L]}=\frac{\zeta_{\Lambda}(x_1)^{\frac{x_2}{x_2-x_1}} \zeta_{L}(x_2)^{\frac{x_1}{x_2-x_1}}  }{\zeta_{\Lambda}(x_2)^{\frac{x_1}{x_2-x_1}} \zeta_{L}(x_1)^{\frac{x_2}{x_2-x_1}} }.
\end{equation}
We have, for any $L\in \mathcal{L}_d^1$,
\begin{align}
\displaystyle \lim_{x_2\to +\infty} \zeta_L(x_2)^{\frac{x_1}{x_2-x_1}}&=\lim_{x_2\to +\infty}\left( \tau(L) +\sum_{p\in L^{(>1)}}|p|^{-x_2}  \right)^{\frac{x_1}{x_2-x_1}}\\
&=\lim_{x_2\to +\infty} \tau(L)^{\frac{x_1}{x_2-x_1}}\left( 1+\tau(L)^{-1}\sum_{p\in L^{(>1)}}|p|^{-x_2} \right)^{\frac{x_1}{x_2-x_1}}=1\ .
\end{align}
It follows that, using the fact that we have the normalization $L,\Lambda\in \mathcal L_d^1$,
\begin{align*}
\lim_{x_1\to +\infty}\lim_{x_2\to +\infty}\frac{\displaystyle \min_{\lambda>0}E_f[\lambda\Lambda]}{\displaystyle \min_{\lambda>0}E_f[\lambda L]} &=\lim_{x_1\to +\infty} \frac{\zeta_{\Lambda}(x_1)}{\zeta_L(x_1)}\\
&=\lim_{x_1\to +\infty}  \frac{\tau(\Lambda) + \sum_{p\in\Lambda^{(> 1)}} |p|^{-x_1}}{\tau(L) + \sum_{p\in L^{(>1)}} |p|^{-x_1}}=\frac{\tau(\Lambda)}{\tau(L)}.
\end{align*}
The final statement in the proposition holds by noting that $\min_{\lambda>0}E_f[\lambda\Lambda]<0$ as follows by taking $\lambda$ large enough, and by then using Remark~\ref{rmk:kissing}.
\end{proof}
\begin{rmk}
The optimality of the FCC lattice for large exponent proved above is consistent with \cite[Conjecture 1.7]{Beterminlocal3d} which in the absolutely summable case states that optimality of the $FCC$ lattice should hold in dimension $3$ for any exponents $x_2>x_1>3$.
\end{rmk}
\begin{rmk}[Testing the Theil crystallization criterion on Lennard-Jones-type potentials]
We note that in \cite[Th. 1.1]{Crystal} there were given conditions for a potential $V$ with one well, and needed to produce crystallization to a triangular lattice. Although these conditions are robust enough to accomodate a potential which is of power-law type near $r=0$ and near $r=+\infty$, these conditions on the other hand are not weak enough to include the case of Lennard-Jones type potentials of the form considered here. Indeed, if we normalize our potentials as indicated in \cite{Crystal}, so that $f(1)=-1$ is the minimum of $f$, namely we look at 
\[
f_{x_1,x_2}(r):=\frac{x_1}{x_2-x_1}\frac{1}{r^{x_2}} - \frac{x_2}{x_2-x_1}\frac{1}{r^{x_1}}\ ,
\]
then we find numerically the conditions $f_{x_1,x_2}''(1+\alpha)\ge 1$ and $f_{x_1,x_2}(1+\alpha)\ge -\alpha$, required as sufficient conditions in \cite[(3), (4)]{Crystal} for suitable $\alpha\in(0,\alpha_0)$ (where the upper bound $\alpha_0<1/3$ is implicitly found during the proof of \cite[Thm. 1.1]{Crystal}, and needs in fact to be fixed as a positive number much smaller than $1/3$). What we found is that the above conditions on $f_{x_1,x_2}$ are never simultaneously achieved for any choice of $\alpha<1/3$ for powers $2<x_1< x_2$.
\end{rmk}

\section{Some counterexamples}\label{sect-contrex}
\subsection{The triangular lattice can be minimizer at all scales for non-completely monotone $f$}\label{not_cm}
In the following example, we have numerically checked that there exist functions which are not completely monotone and such that the triangular lattice is the minimizer of $E_f$ at all scales.

\medskip

For any $\varepsilon>0$, define $f_\varepsilon:(0,+\infty)\to \R$ by
\begin{equation}
f_\varepsilon(r):=\frac{6}{r^4}-\frac{2(2+\varepsilon)}{r^3}+\frac{1+\varepsilon}{r^2}.
\end{equation}
The inverse Laplace transform of $f_\varepsilon$ is $\mathcal{L}^{-1}[f_\varepsilon](x)=x(x-1)(x-1-\varepsilon)$, and in particular $\mathcal{L}^{-1}[f_\varepsilon](x)\leq 0$ on $[1,1+\varepsilon]$. In Figure \ref{fig-lapinvepsilon} and Figure \ref{fig-graphepsilon}, we have respectively plotted $\mathcal{L}^{-1}[f_\varepsilon]$ and $f_\varepsilon$ for different values of $\varepsilon$.

\begin{figure}[!h]
\centering
\includegraphics[width=11cm]{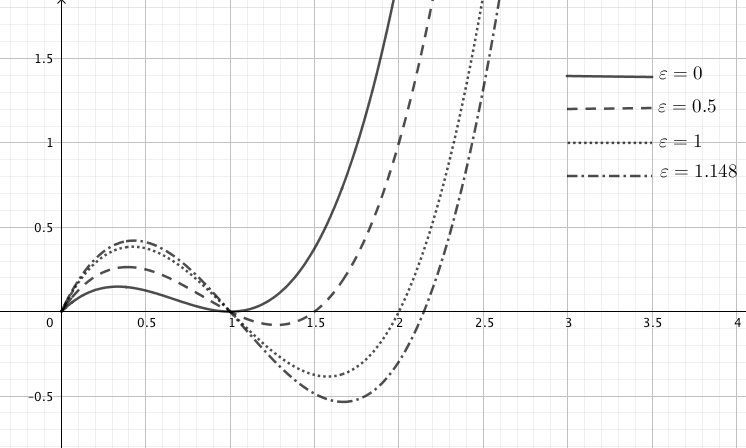} \\
\caption{Plot of function $\mathcal{L}^{-1}[f_\varepsilon]$ for $\varepsilon\in \{0,0.5,1,1.148\}$}
\label{fig-lapinvepsilon}
\end{figure}
\begin{figure}[!h]
\centering
\includegraphics[width=11cm]{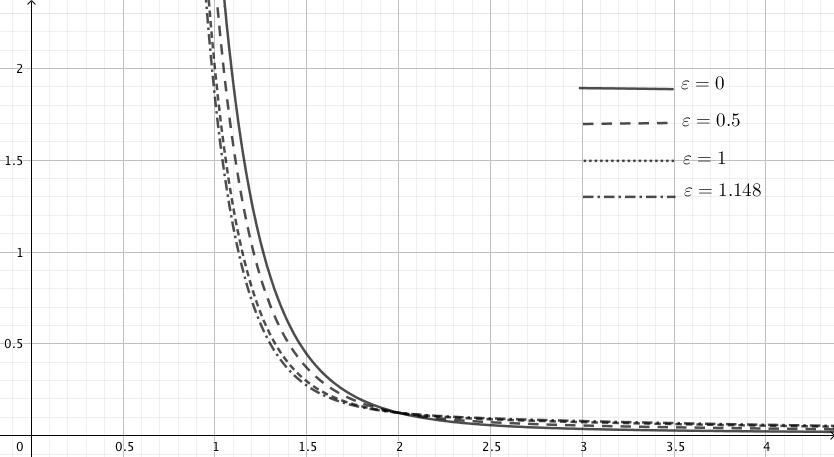} \\
\caption{Plot of function $f_\varepsilon$ for $\varepsilon\in \{0,0.5,1,1.148\}$}
\label{fig-graphepsilon}
\end{figure}
We now justify the following claim, based on numerical evidence:
\begin{claim}
There exists $\varepsilon_0\approx 1.148>0$ such that for any $0\leq \varepsilon<\varepsilon_0$ and any $\lambda>0$, $\Lambda_1$ is the unique minimizer of $L\mapsto E_{f_\varepsilon}[\lambda L]$ on $D_{\mathcal{L}_2^\circ}$.
\end{claim}
We have, using the scaling property of the Epstein zeta function, for any Bravais lattice $L\in D_{\mathcal{L}_2^\circ}$
\begin{equation}\label{difference}
 E_{f_\varepsilon}[\lambda L]-E_{f_\varepsilon}[\lambda \Lambda_1]=\frac{P_{\varepsilon,L}(\lambda^2)}{\lambda^8}
\end{equation}
where, using the notation $d_s(L):=\zeta_{L}(s)-\zeta_{\Lambda_1}(s)$, we have

\[P_{\varepsilon,L}(X):=(1+\varepsilon)d_4(L)X^2 -2(2+\varepsilon)d_6(L)X+6d_8(L).
\]

The discriminant of polynomial $P_{\varepsilon, L}$ is

\begin{equation}\label{delta_eps_theta}
\Delta(\varepsilon,L)=4(2+\varepsilon)^2d_6(L)^2-24(1+\varepsilon)d_4(L)d_8(L).
\end{equation}

Note that, for all any $s>2$, $d_s(L)\ge 0$ with equality only for $L=\Lambda_1$ on $D_{\mathcal{L}_2^\circ}$. Furthermore, we have that $\Delta(\varepsilon,L)< 0$ for $\varepsilon\geq 0$ and $L\neq \Lambda_1$ (on $D_{\mathcal{L}_2^\circ}$) if and only if
\begin{equation}
h(\varepsilon):=\frac{(2+\varepsilon)^2}{6(1+\varepsilon)}< \frac{d_4(L)d_8(L)}{d_6(L)^2}=:c(L).
\end{equation}
Note that $h$ is an increasing function on $[0,+\infty)$ and $h(0)=2/3$. Furthermore, we can prove  that $\lim_{L\mapsto \Lambda_1} c(L)=c\in \R$. Thus, if we find $\varepsilon_0$ such that
\begin{equation}
h(\varepsilon_0)< \inf_{L, |L|=1} c(L),
\end{equation}
then $\Delta(\varepsilon,L)< 0$ for any $0\leq \varepsilon \leq \varepsilon_0$ and any $L\in D_{\mathcal{L}_2^\circ}$.

\medskip

As in \cite[Sect. 3]{Coulangeon:kx}, we recall that for $L\in\mathcal L^\circ_2$, we can write
\[
\zeta_L(s)=\sum_{p\in L\setminus\{0\}}\frac{1}{|p|^s}=\sum_{x\in\Z^2\setminus\{0\}}\frac{1}{(x^tAx)^{s/2}}\ ,
\]
for suitable $A\in\mathcal{P}_2^\circ$. We can view $\mathcal P_2^\circ$ as a differential submanifold of the vector space of $2\times2$ symmetric matrices with real entries $\mathcal{S}_2(\R)$. The tangent space to $\mathcal P_2^\circ$ at any point $A\in\mathcal S_2(\R)$ then identifies with the set $\{H\in \mathcal{S}_2(\R):\ Tr(A^{-1}H)=0\}$. Moreover, the exponential map $H\mapsto e_A(H):=A \exp(A^{-1}H)$ induces a local diffeomorphism from the tangent space to $\mathcal{P}_2^\circ$. Noting that for $\Lambda_1$ all layers hold a $4$-design (see \cite[Def. 6.3, Thm. 6.12]{venkov}), we therefore obtain, by \cite[Eq. (3.5), with dimension $n=2$ and exponent $s/2$ in our case]{Coulangeon:kx}, that for $A$ corresponding to the lattice $\Lambda_1$ and for $L$ another lattice corresponding to the quadratic form with matrix $e_A(H)$ there holds
\[\label{TaylorEpstein}
\zeta_L(s)-\zeta_{\Lambda_1}(s)= \frac{s(s-2)}{32}\op{Tr}(A^{-1}H)^2\zeta_{\Lambda_1}(s) + o\left(\|H^2\|\right)\ .
\]
Thus we have
\[
c(L)=\frac{(8 Tr(A^{-1}H)^2\zeta_{\Lambda_1}(4)+o(\|H^2\|))(48\ Tr(A^{-1}H)^2\zeta_{\Lambda_1}(8)+o(\|H^2\|))}{(24\ Tr(A^{-1}H)^2\zeta_{\Lambda_1}(6)+o(\|H^2\|))^2}.
\]
And we then get, taking $H\to 0$, $H\in\mathcal S_2(\R)\setminus\{0\}$, 
\begin{align}
\lim_{L\to \Lambda_1}c(L)=\frac{2}{3}\frac{\zeta_{\Lambda_1}(4)\zeta_{\Lambda_1}(8)}{\zeta_{\Lambda_1}(6)^2}\approx 0.7719234>\frac{2}{3}.
\end{align}
Numerically, we find that
\begin{equation}
\inf_{L, |L|=1} c(L) > 0.769.
\end{equation}

To obtain the above, we have proceeded as follows. We parametrized $L\in D_{\mathcal L_2^\circ}$ as $L(x,y)$ with $(x,y)$ a point of the half elliptic fundamental domain
$$
\tilde D=\{(x,y)\in \R^2 : 0\leq x\leq 1/2, x^2+y^2\geq 1  \},
$$
and denote by slight abuse of notation $c(L(x,y))=c(x,y)$. We have first computed the limit of $c(x,y)$ for fixed $x\in [0,1/2]$ as $y\to +\infty$. Using the fact that $\zeta_L(s)=2y^{\frac{s}{2}}\zeta(s)+h(s,y)$ where $h(s,y)=o(y^{\frac{s}{2}})$ (see \cite[Thm 1]{katsurada}) and the well-known exact values of $\zeta(2k)$ for $k\in \{2,3,4\}$ of the Riemann zeta function $\zeta(s)=\sum_{m>0}m^{-s}$, which are
$$
\zeta(4)=\frac{\pi^4}{90},\quad \zeta(6)=\frac{\pi^6}{945},\quad \zeta(8)=\frac{\pi^8}{9450},
$$
we obtain, for any fixed $x\in [0,1/2]$,
\begin{equation}
\lim_{y\to +\infty} c(x,y)=\lim_{y\to +\infty} \frac{\zeta_L(4)\zeta_L(8)}{\zeta_L(6)^2}=\frac{\zeta(4)\zeta(8)}{\zeta(6)^2}=\frac{21}{20}=1.05.
\end{equation}
It is therefore possible to reduce the research of the minimizer of $c$ to a compact subset of $\tilde{D}$. A numerical study of $c$ shows that, for $y>3$, $c(x,y)>0.9$. Therefore, we have computed the values of $c(x,y)$ for $(x,y)$ on a grid and on the set $\{(x,y)\in \tilde{D} : x^2+y^2=1\}$, i.e. among rhombic lattices. The results can be seen in Figure \ref{fig-c3}.

\begin{figure}[!h]
\centering
\includegraphics[width=6cm]{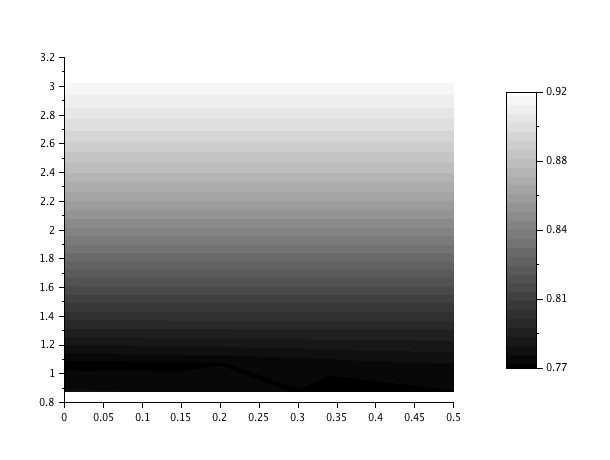}  \quad \includegraphics[width=6cm]{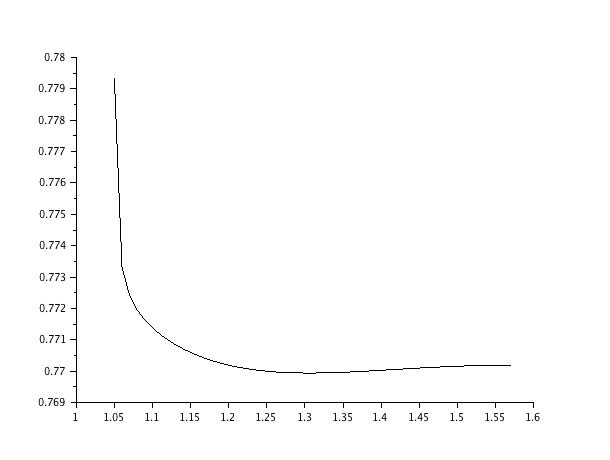}
\caption{Plot of $c(x,y)$ for $\sqrt{3}/2<y<3$ (on the left) and $c(\cos(t),\sin(t))$, $t\in \left[\frac{\pi}{3},\frac{\pi}{2} \right]$ (on the right).}\label{fig-c3}
\end{figure}

We finally get, from an enough refined mesh of the values $(x,y)\in \tilde{D}$ such that $y\leq 3$, that $\min_{(x,y)\in \tilde{D}} c(x,y)>0.769$ and this minimum is approximatively equal to $0.7699393$, achieved very close to the point $(x,y)=(0.2607474,0.9654071)$.

\medskip

Hence, solving $h(\varepsilon)\leq 0.769$, we get $\varepsilon_0\approx 1.1485753<1.148$. Therefore, we have numerically checked that, for any $0\leq \varepsilon \leq 1.148$ and any Bravais lattice $L\in \mathcal{L}_2^\circ$, $\Delta(\varepsilon,L)< 0$, i.e., for any $\lambda>0$ and any $\varepsilon,L$ as above,
\begin{equation}
 E_{f_\varepsilon}[\lambda L]-E_{f_\varepsilon}[\lambda \Lambda_1]\geq 0,
\end{equation}
with equality if and only if $L=\Lambda_1$.

\begin{rmk}
We have checked that, for $\varepsilon=\varepsilon_0\approx 1.148$, $f_\varepsilon$ satisfies
$$
\forall r>0,\quad \forall k\leq 4,\quad (-1)^k f_\varepsilon^{(k)}(r)\geq 0,
$$
and $f_\varepsilon^{(5)}$ is not non-positive on $(0,+\infty)$.
\end{rmk}

\subsection{Potentials with one (strict) well such that the global minimum is not triangular}\label{subsec-onewell}
In this section we provide an abstract result (in Proposition~\ref{prop:theil_ZA}) in which the existence of a whole class of such potentials (in particular, as is easy to verify by Remark \ref{rmk:existza}, we can construct such potentials with arbitrary smoothness). We present first an explicit example of a Lipschitz potential for which the square lattice is favored over the triangular one.

\medskip

The ease to create examples, even under the constraint of having precisely one well, and even at a high level of smoothness, suggests that even beyond the case of the square lattice, the possibilities to construct potentials that favor a given lattice over the triangular one seem to be relatively vast. At the end of this section, we present a remark on how to do this for more general lattices (see Remark \ref{rmk-generalprinciple}).
Inspired by a counter-example by Ventevogel \cite[Sec. 5]{VN1}, we define the following potential, which gives another more explicit example:
$$
g(r):= \left\{ \begin{array}{ll}
\displaystyle 30000 & \mbox{if $0<r<4/9$}\\
2-3r & \mbox{if $4/9\leq r\leq 1$}\\
-r^{-4} & \mbox{if $r>1$,}\\
\end{array}\right.
$$
This function is clearly strictly decreasing on $(0,1)$ and strictly increasing on $(1,+\infty)$. Defining $f$ by $f(r^2)=g(r)$, we have numerically computed the following quantities:
$$
\min_{\lambda>0} E_f[\lambda \Z^2]\approx -8.5915114 , \quad \min_{\lambda >0} E_f[\lambda\mathsf{A}_2]\approx -7.7107743.
$$
We therefore can see that $\min_{\lambda >0} E_f[\lambda \Z^2]< \min_{\lambda>0} E_f[\lambda \mathsf{A}_2]$, and that (numerically) shows that a global minimizer of $E_f$ cannot be a triangular lattice $\lambda \mathsf{A}_2$.

\medskip

We then can construct the following continuous potential $g$, equal to the above one for $r\geq 4/9$ and with strong repulsion near the origin, which satisfies the same property of non-minimality for triangular lattices. The proof of this proposition is given in Appendix \ref{appendix}.

%
%
%\begin{figure}[!h]
%\centering
%\includegraphics[width=9cm]{CE1.png} \\
%\caption{Plot of function $f$ defined in Proposition \ref{prop:example_vn}}
%\label{CE1}
%\end{figure}

\begin{proposition}\label{prop:example_vn}
Let $g$ be the continuous potential defined by
$$
g(r):= \left\{ \begin{array}{ll}
\displaystyle \frac{\left( \frac{2}{3} \right)\left(\frac{4}{9}  \right)^{p}}{r^{p}} & \mbox{if $0<r<4/9$}\\
2-3 r & \mbox{if $4/9\leq r\leq 1$}\\
-r^{-4} & \mbox{if $r>1$,}\\
\end{array}\right.
$$
Then there exists $p_0$ such that for any $p>p_0$ we have for $f(r^2)=g(r)$,
\begin{equation}\label{min_za}
\displaystyle \min_{\lambda>0} E_f\left[\lambda \Z^2\right]<\min_{\lambda>0} E_f[\displaystyle \lambda \mathsf{A}_2].
\end{equation}
\end{proposition}
Furthermore, we finally prove the existence of a non-countable family of $C^1$-functions $f(r^2)=g(r)$ for which the global minimizer of $E_f$ cannot be triangular.

\begin{proposition}\label{prop:theil_ZA}
For any choice of parameters $0<\alpha_0<\alpha_1<1<\sqrt 3 \alpha_0$, there exist $C_0,C_1>0$ such that if a $C^1$-function $g:[0,+\infty)\to\R\cup\{+\infty\}$ is such that $g(|x|), g^\prime(|x|)$ are integrable on $\mathbb R^2\setminus B_\epsilon$ for $0<\epsilon\le\min\{\alpha_1/2,\sqrt 3\alpha_0/2\}$, and satisfies the conditions 
\begin{subequations}\label{imposed}
\begin{equation}
 g(r)\ge 0\quad\mathrm{for}\quad r\in(0,\alpha_1],\quad\mathrm{and}\quad g(r)\le 0\quad\mathrm{for}\quad r\in[1,+\infty),\label{pos}
\end{equation}
\begin{equation}
g'(r)\ge 0\quad\mathrm{for}\quad r\in(\sqrt 3\alpha_0,+\infty),\label{incr}
\end{equation}
\begin{equation}
\min_{r>0}g(r)=g(1)=-1,\label{normalize}
\end{equation}
\begin{equation}
\forall r\in(0,\alpha_0],\quad g(r)\ge -\frac{C_0}{r^2}\int_{\alpha_1}^{+\infty} g(\rho)\rho\mathrm{d}\rho+ \frac{C_0}{r}\int_{\alpha_1/2}^{+\infty}|g^\prime(\rho)|\rho \mathrm{d}\rho,\label{hardwall}
\end{equation}
\begin{equation}
\min_{\alpha_0\le r\le1}\left(g(r)+g\left(\sqrt{2} r\right)\right)\le -\frac32 +\frac14 E_f\left[\alpha_0 \mathsf{A}_2^{(\ge 2)}\right] ,\label{zbettera}
\end{equation}

\end{subequations}
then, for $f(r^2)=g(r)$, and $\lambda_0^{\mathbb Z^2},\lambda_0^{\mathsf{A}_2}\in[\alpha_0,1]$, we have
$$
\displaystyle \min_{\lambda>0} E_f\left[\lambda \Z^2\right]= E_f\left[\lambda_0^{\Z^2}\Z^2\right]<E_f\left[\lambda_0^{\mathsf{A}_2}\mathsf{A}_2\right]=\min_{\lambda>0} E_f[\displaystyle \lambda \mathsf{A}_2].
$$
\end{proposition}

\begin{rmk}[discussion of the conditions \eqref{imposed}]\label{rmk:existza}
The sign conditions \eqref{pos} and \eqref{incr} could possibly be relaxed, at the cost of complicating the proof, and in the present form they are already true for a large class of one-well potentials, while \eqref{normalize} is just a normalization condition. 

\medskip

Condition \eqref{zbettera} seems to be the most restrictive one, but we can construct a $g$ that satisfies it by the following rough procedure. First, we ensure that there exists $\alpha \in(\alpha_0,1]$ such that $g(\alpha)+g(\sqrt 2\alpha)\le -3/2-\epsilon$. Then, we make $|g(r)|$ very small compared to $\epsilon$ in the range $r>\sqrt 3\alpha_0$.

\medskip

If $g(|x|), g^\prime(|x|)$ are integrable on $\mathbb R^2$ away from the origin, then condition \eqref{hardwall} is true if $g(r)\ge \tilde C/r^2$ in a neighborhood of the origin. The role of this condition is to form an ``effective hard core'' for our interactions: as we will see in the proof of the proposition, it implies that the optimal value from \eqref{min_za} is larger than $\alpha_0$.

\medskip

Then $\alpha_1$ can be defined to be the only value where $g(\alpha_1)=0$, therefore does not represent a constraint on the possible choices of $g$, and we merely introduced the notation for it in the statement for the sake of explicitly representing the behavior of $g$.
\end{rmk}

\begin{proof}[Proof of Proposition~\ref{prop:theil_ZA}:] Note that both lattices $\Lambda=\mathbb Z^2, \mathsf{A}_2$ have as first shell a set of points ($4$ points for $\Lambda=\mathbb Z^2$ and $6$ points for $\Lambda=\mathsf{A}_2$) at distance $1$ from the origin. The proof consists in comparing the contributions from this first shell to $E_f[\lambda \Lambda]$ to the contributions from all the remaining shells. While $\mathbb Z^2$ has fewer points in the first shell, the second shell of $\mathbb Z^2$ is closer to the origin, at distance $\sqrt 2$ from the origin for $\Lambda=\mathbb Z^2$ and it is farther from the origin, at distance $\sqrt 3$ from the origin, for $\Lambda=\mathsf{A}_2$. 

\medskip 

The principle of the proof is the following: the above properties of $g$ are such that the effect of the disparity of the contributions from the second shell (which is advantageous for $\mathbb Z^2$) ``wins'' over the effect of the disparity of the first shell (which itself would be more advantageous to $\mathsf{A}_2$). We first show that $\lambda_0^{\mathbb Z^2}, \lambda_0^{\mathsf{A}_2}\in[\alpha_0,1]$, in steps 1 and 2, after which in step 3 we use the above specific discussion on shells to conclude.

\medskip

\textbf{Step 1.} We show that for $\Lambda\in\{\mathbb Z^2,\mathsf{A}_2\}$ and $\lambda\in ]0,\alpha_0]$ there holds $E_f[\lambda\Lambda]\ge 0$. This will follow by \eqref{pos} and from the bound
\begin{equation}\label{energypositive}
E_f[\lambda\Lambda\cap B_{\alpha_0}]\ge - E_f[\lambda\Lambda\setminus B_{\alpha_1}],
\end{equation}
which we are proving now.

\medskip

We use a rough quadrature estimate in order to bound the left hand side of \eqref{energypositive} as follows. By Poincar\'e inequality (which holds for bounded convex domains such as $\mathcal V_{\lambda\Lambda}$)
\begin{equation}\label{taylor}
\frac{1}{|\mathcal V_{\lambda \Lambda}|}\int_{\mathcal V_{\lambda \Lambda}(p)}\left|g(|p|)-g(|x|)\right|\mathrm{d}x\le C\frac{\mathrm{diam}\mathcal V_{\lambda\Lambda}}{|\mathcal V_{\lambda \Lambda}|} \int_{\mathcal V_{\lambda\Lambda}(p)}|g'(|x|)|\mathrm{d}x,
\end{equation}
where for the Voronoi cells of $\lambda\Lambda$ we use the notation $\mathcal V_{\lambda\Lambda}(p)=\{x\in\mathbb R^2: |x-p| \leq\mathrm{dist}(x, \lambda\Lambda)\}$ and $\mathcal V_{\lambda \Lambda}=\mathcal V_{\lambda \Lambda}(0)$. By summing formula \eqref{taylor} over $p\in\lambda\Lambda\setminus B_{\alpha_1}$, we find the bound
\begin{multline}\label{quadbound}
\left|\frac{1}{|\mathcal V_{\lambda\Lambda}|}\int_{\bigcup_{p\in\lambda\Lambda\setminus B_{\alpha_1}}\mathcal V_{\lambda\Lambda}(p)}g(|x|)\mathrm{d}x - \sum_{p\in\lambda\Lambda\setminus B_{\alpha_1}}g(|p|)\right| \\
\le C\frac{\mathrm{diam}\mathcal V_{\lambda\Lambda}}{|\mathcal V_{\lambda\Lambda}|}\int_{\mathbb R^2\setminus B_{\alpha_1/2}}|g^\prime(r)|r \mathrm{d}r,
\end{multline}
where in the last integral we used the fact that for $\lambda<\alpha_1$, Voronoi cells of $\lambda\Lambda$ corresponding to points outside $B_{\alpha_1}$ do not intersect $B_{\alpha_1/2}$. The bound \eqref{quadbound} gives the following control on the right-hand side of \eqref{energypositive}, in which the constants $C_\Lambda>0$ only depend on the dimension and on the shape of $\mathcal V_{\lambda\Lambda}$:
\begin{equation}\label{controlright}
-E_f[\lambda\Lambda\setminus B_{\alpha_1}]\le -\frac{C_\Lambda}{\lambda^2}\int_{\alpha_1}^{+\infty} g(r)r\mathrm{d} r +\frac{C_\Lambda}{\lambda}\int_{\alpha_1/2}^{+\infty}|g^\prime(r)|r\mathrm{d}r.
\end{equation}
Now in order to pass from \eqref{controlright} to \eqref{energypositive} it suffices to note that for $\lambda<\alpha_0$ and $\Lambda\in\{\Z^2,\mathsf{A}_2\}$ the first shell of $\Lambda$ is at distance $\lambda<\alpha_0$ from the origin. Thus we may use the version of the bound \eqref{hardwall}, together with \eqref{controlright}: we obtain that \eqref{energypositive} is true, provided that the constant $C_0>C_\Lambda/\#(\Lambda^{(1)})$, where $\Lambda^{(1)}\subset\Lambda$ is the first shell in $\Lambda$ and $C_\Lambda$ is the constant from \eqref{controlright}.

\medskip

Now due to \eqref{energypositive} we have that $ E_f[\lambda\Lambda]\ge 0$ whenever $\lambda<\alpha_0$, whereas due to \eqref{pos} we have that $\min_{\lambda>0} E_f[\lambda\Lambda]\le E_f[\Lambda]<0$. This shows that the value $\lambda_0^\Lambda$ at which $\min_{\lambda>0} E_f[\lambda\Lambda]$ is achieved, does not lie in the interval $]0,\alpha_0[$.

\medskip

\textbf{Step 2.} Now we use condition \eqref{incr} on $g^\prime$, which implies that for $\lambda\in[1,+\infty)$, the sum over the shells different than the first is increasing in $\lambda$, whereas \eqref{normalize} shows that the sum over first shell is minimized at $\lambda = 1$, showing that the case $\lambda_0^\Lambda>1$ is also not possible. This shows that the minimum of $\lambda\mapsto  E_f[\lambda\Lambda]$ occurs at $\lambda_0^\Lambda\in[\alpha_0,1]$ for our lattices $\Lambda\in \{\mathbb Z^2,\mathsf{A}_2\}$.

\medskip

\textbf{Step 3.} Now that we know that the minimum of $\lambda\mapsto E_f[\lambda\Lambda]$ is achieved for $\lambda_0^\Lambda$ belonging to the interval $[\alpha_0,1]$ for both choices $\Lambda\in\{\Z^2,\mathsf{A}_2\}$, we need only to check the validity of the inequality 
\begin{equation}\label{boundza_rough}
E_f\left[\lambda_0^{\mathsf{A}_2} \mathsf{A}_2\right]> E_f\left[\lambda_0^{\mathbb Z^2} \mathbb Z^2\right].
\end{equation}
For simplicity of notation, we simply denote $\lambda_1=\lambda_0^{\mathsf A_2}, \lambda_2:=\lambda_0^{\mathbb Z^2}$ for the rest of the proof.

\medskip

We note that due to the constraints $\lambda_1,\lambda_2\in[\alpha_0,1]$ we have that 
\begin{itemize}
\item $\left(\lambda_1 \mathsf{A}_2\right)^{(\ge 2)}\subset \mathbb R^2\setminus B_{\sqrt 3 \alpha_0}$,
\item $\left(\lambda_2 \mathbb Z^2\right)^{(\ge 3)}\subset \mathbb R^2\setminus B_{2 \alpha_0}$.
\end{itemize}
This has the following two consequences: Firstly, due to the fact that $1<\sqrt 3\alpha_0<2\alpha_0$ and the second sign condition in \eqref{pos}, we have 
\begin{equation}\label{negative_tail}
E_f\left[\left(\lambda_1 \mathsf{A}_2\right)^{(\ge 2)}\right]<0, \quad E_f\left[\left(\lambda_2 \mathbb Z^2\right)^{(\ge 3)}\right]<0.
\end{equation}
Secondly, due to the condition \eqref{incr}, we have that $E_f\left[\left(\lambda \mathsf{A}_2\right)^{(\ge 2)}\right]$ is increasing in $\lambda$ for $\lambda>\sqrt 3\alpha_0$, therefore in particular
\begin{equation}\label{min_tail_A}
\min_{\lambda\in[\alpha_0,1]}E_f\left[\left(\lambda \mathsf{A}_2\right)^{(\ge 2)}\right]=E_f\left[\left(\alpha_0 \mathsf{A}_2\right)^{(\ge 2)}\right].
\end{equation}
Thus, e can use \eqref{normalize}, \eqref{zbettera}, \eqref{negative_tail} and \eqref{min_tail_A} to prove the following inequalities:
\begin{eqnarray}
\lefteqn{\sum_{p\in (\lambda_1\mathsf{A}_2)^{(1)}}g(|p|) - \sum_{p\in(\lambda_2\mathbb Z^2)^{(1,2)}}g(|p|)=6g(\lambda_1)-4g(\lambda_2)-4g\left(\sqrt 2 \lambda_2\right)}\nonumber\\
&\stackrel{\eqref{normalize}}{\ge}&-4g(\lambda_2)-4g\left(\sqrt 2 \lambda_2\right)-6\nonumber\\
&\stackrel{\eqref{zbettera}}{>}& - E_f\left[\left(\alpha_0\mathsf{A}_2\right)^{(\ge 2)}\right]\nonumber\\
&\stackrel{\eqref{negative_tail}, \eqref{min_tail_A}}{>}& -E_f\left[\left(\lambda_1\mathsf{A}_2\right)^{(\ge 2)}\right] + E_f\left[\left(\lambda_2\mathbb Z^2\right)^{(\ge 3)}\right],\label{boundza}
\end{eqnarray}
Now the last line in \eqref{boundza} by reordering terms gives \eqref{boundza_rough}, and completes the proof of \eqref{min_za}.
\end{proof}
\begin{rmk}
In \cite[Fig. 1]{Crystal}, Theil numerically noticed that $\min_{\lambda>0} E_f[\lambda \Z^2]<\min_{\lambda>0} E_f[\lambda \mathsf{A}_2]$ for $g(r)=\frac{1}{r^{12}}+\tanh\left(4r-\frac{13}{2}  \right)-1$, $f(r^2)=g(r)$, based on a similar principle as in our above theorem. This potential is not a one-well potentinal, as it is decreasing at infinity. Potentials with several wells, and possibly with the property of being decreasing at infinity (property which could be interpreted as having ``a well at infinity''), are not discussed in the present paper.
\end{rmk}
\begin{rmk}[A general principle]\label{rmk-generalprinciple}
Note that the proof of Proposition \ref{prop:theil_ZA} is relatively robust, and allows to introduce also perturbations of the square lattice, and to favor them over the triangular lattice. In a setting of arbitrary dimension and for arbitrary lattices, we formulate the following principle: If for two lattices $\Lambda_1,\Lambda_2$ there holds
\begin{equation}
|\Lambda_1|=|\Lambda_2|\quad\mbox{and}\quad\exists\ r>0,\quad\#\left(\Lambda_1\cap B(0,r)\right)>\#\left(\Lambda_2\cap B(0,r)\right),
\end{equation}
then it is possible to construct a potential $f$ with one single well, such that \eqref{min_za} holds with the lattices $\mathsf{A}_2,\mathbb Z^2$ replaced by $\Lambda_1,\Lambda_2$, respectively. This can be done, amongst other methods, by imitating the proof of Proposition \ref{prop:theil_ZA}.
\end{rmk}

\paragraph{\textbf{Acknowledgements.}} The authors wish to thank the anonymous referee for helping clarify the paper.

LB is grateful for the support of the Mathematics
Center Heidelberg (MATCH) during his stay in Heidelberg. He also acknowledges support
from ERC advanced grant Mathematics of the Structure of Matter (project
No. 321029) and from VILLUM FONDEN via the QMATH Centre of Excellence (grant
No. 10059). MP is grateful for the stimulating work environment provided by ICERM (Brown University), during the Semester Program on ``Point Configurations in Geometry, Physics and Computer Science'' in spring 2018 supported by the National Science Foundation under Grant No. DMS-1439786, and acknowledges support from the FONDECYT \emph{Iniciacion en Investigacion 2017} grant N. 11170264. 

\appendix

\section{Proof of Proposition \ref{prop:example_vn}}\label{appendix}
\begin{proof}[Proof of Proposition \ref{prop:example_vn}:]
In the following, we will write $c=\left( \frac{2}{3} \right)\left(\frac{4}{9}  \right)^{p}$ and $I_\lambda=\left[\frac{4}{9\lambda},\frac{1}{\lambda}  \right]$.
Let $L\in \mathcal{L}_d$ be a Bravais lattice, then we have, for any $\lambda>0$,
\begin{equation}\label{formula_egl}
E_f[\lambda L]=\frac{c}{\lambda^{p}}\sum_{x\in L \atop |x|<\frac{4}{9\lambda}}\frac{1}{|x|^{p}}+\sum_{x\in L \atop |x|\in I_\lambda}(2-3\lambda |x|)-\frac{1}{\lambda^4}\sum_{x\in L \atop |x|>\frac{1}{\lambda}} \frac{1}{|x|^4}:=S_1+S_2+S_3.
\end{equation}
\textbf{Step 1.} We remark that, for any $L$ with minimal distance $1$ (like $\Z^2$ and $\mathsf{A}_2$), therefore if $\lambda>1$, then $S_1=S_2=0$ and we get $E_f[\lambda L]=S_3=-\lambda^{-4} \zeta_{L}(4)$. Since $\lambda\mapsto \lambda^{-4}$ is decreasing, it follows that $E_f[\lambda L]$ is increasing in $\lambda$ for $\lambda\in(1,+\infty)$.

\medskip

\textbf{Step 2.} We now treat the case $L=\Z^2$, for $\lambda\in[4/9,1]$, a range in which $S_1=0$ but $S_2,S_3\neq 0$. Note that the distances to the origin for the square lattice are $1$ (achieved $4$ times), $\sqrt{2}$ (achieved $4$ times), $2$ (achieved $4$ times) and $\sqrt{5}$ (achieved $8$ times). We base our case subdivision on these values.
\begin{enumerate}\item \textit{Values $4/9\leq \lambda \leq 1/\sqrt{5}$ for $L=\Z^2$.} Then
\begin{equation}
E_f[\lambda \Z^2]=40-3(12+4\sqrt{2}+8\sqrt{5})\lambda -\frac{(\zeta_{\Z^2}(4)-5-1/4-8/25)}{\lambda^4}.
\end{equation}
Therefore, $\frac{d}{d\lambda}E_f[\lambda \Z^2]\geq 0$ if and only if $\lambda\leq \left(\frac{4(\zeta_{\Z^2}(4)-5-1/4-8/25)}{3(12+4\sqrt{2}+8\sqrt{5})}  \right)^{1/5}=:\lambda_1\approx 0.4433$. Thus, $\lambda\mapsto E_f[\lambda \Z^2]$ is decreasing on $[4/9,1/\sqrt{5}]$.

\item \textit{Values $1/\sqrt{5}< \lambda\leq 1/2$ for $L=\Z^2$.} In this case
\begin{equation}\label{case2}
E_f[\lambda \Z^2]=24-3(12+4\sqrt{2})\lambda -\frac{(\zeta_{\Z^2}(4)-5-1/4)}{\lambda^4}.
\end{equation}
The $\tfrac{d}{d\lambda}$-derivative of \eqref{case2} is positive for $ \lambda<\left(\frac{4(\zeta_{\Z^2}(4)-5-1/4)}{3(12+4\sqrt{2})}  \right)^{1/5}\approx 0.56$ and thus $\lambda\mapsto E_f[\lambda \Z^2]$ is increasing for $\lambda\in(1/\sqrt{5},1/2]$.

\item \textit{Values $1/2<\lambda\leq 1/\sqrt{2}$ for $L=\Z^2$.} In this case
\begin{equation}
E_f[\lambda \Z^2]=16-3(4+4\sqrt{2})\lambda -\frac{(\zeta_{\Z^2}(4)-5)}{\lambda^4}.
\end{equation}
Now for the critical value $\left(\frac{4(\zeta_{\Z^2}(4)-5)}{3(4+4\sqrt{2})}  \right)^{1/5}=:\lambda_2\approx 0.6765$ we find that $\lambda\mapsto E_f[\lambda \Z^2]$ is increasing on $(1/2,\lambda_2]$ and decreasing on $[\lambda_2,1/\sqrt{2}]$.

\item \textit{Values $1/\sqrt{2}<\lambda\leq 1$ for $L=\Z^2$.} In this case
\begin{equation}
E_f[\lambda \Z^2]=8-12\lambda -\frac{(\zeta_{\Z^2}(4)-4)}{\lambda^4}.
\end{equation}
Therefore, defining $\left( \frac{4(\zeta_{\Z^2}(4)-4)}{12} \right)^{1/5}=:\lambda_3\approx 0.9245$ the map $\lambda\mapsto E_f[\lambda \Z^2]$ is increasing on $[1/\sqrt{2},\lambda_3]$ and decreasing on $[\lambda_3,1]$.
\end{enumerate}
Now, comparing the values of $E_f[\lambda L]$ for $L=\Z^2$ for $\lambda\in \{1/\sqrt{5},1/\sqrt{2},1\}$, we find that, based on the above discussion and on Step 1,
\begin{equation}\label{minsquare}
\displaystyle \min_{\lambda\ge 4/9}E_f[\lambda \Z^2]=E_f\left[\frac{1}{\sqrt{5}}\Z^2  \right]\approx -19.108745
\end{equation}
\textbf{Step 3.} By performing a similar discussion as in Steps 1 and 2, based on the distances to the origin for points in $L=\mathsf{A}_2$ lower than $9/4$, that are $1,\sqrt{3}$ and $2$ (all achieved $6$ times), we obtain
\begin{equation}\label{mintri}
\displaystyle \min_{\lambda\ge 4/9}E_f[\lambda \mathsf{A}_2]=E_f\left[\frac{4}{9}\mathsf{A}_2  \right]\approx -19.013358.
\end{equation}
\textbf{Step 4.} We now assume that $\lambda<\frac{4}{9}$ and we compute a lower bound for the energy \eqref{formula_egl}. We bound $S_1$ by the first term in the sum and $S_3$ by the sum over the whole lattice without the constraint $|x|>\tfrac1\lambda$, and we obtain
\begin{equation}
S_1> \frac{4\left(\frac{2}{3} \right)\left(\frac{4}{9}  \right)^p}{\lambda^p},\quad S_3>-\frac{\zeta_{\Z^2}(4)}{\lambda^4}.
\end{equation}
For $S_2$, we use the fact that $\#\{ x\in \Z^2; |x|\leq r \}=\pi r^2 + R(r)$ where $|R(r)|\leq 2\sqrt{2}\pi r$. We therefore get
 \begin{align}
 S_2=\sum_{x\in \Z^2 \atop |x|\in I_\lambda}(2-3\lambda |x|)&> \left(2-\frac{4}{3\lambda}\right)\#\{x\in \Z^2; |x|\in I_\lambda  \}\\
 &> \left(2-\frac{4}{3\lambda}\right)\left(\frac{\pi}{\lambda^2}-\frac{16\pi}{81\lambda^2}+\frac{2\sqrt{2}\pi}{\lambda}+\frac{8\sqrt{2}\pi}{9\lambda} \right)\\
 &=-\frac{260\pi}{243\lambda^3}-\left(\frac{104\sqrt{2}\pi}{27\lambda^2}-\frac{130\pi}{81\lambda^2} \right)+\frac{52\sqrt{2}\pi}{9\lambda}.
 \end{align}
 Thus, we have obtained
 \begin{equation}
E_f[\lambda \Z^2]>\frac{4\left(\frac{2}{3} \right)\left(\frac{4}{9}  \right)^p}{\lambda^p}-\frac{\zeta_{\Z^2}(4)}{\lambda^4}-\frac{260\pi}{243\lambda^3}-\left(\frac{104\sqrt{2}\pi}{27\lambda^2}-\frac{130\pi}{81\lambda^2} \right)+\frac{52\sqrt{2}\pi}{9\lambda} .
 \end{equation}
We now want to determine a value $\lambda<4/9$ such that $E_f[\lambda \Z^2]>E[\frac{1}{\sqrt{5}}\Z^2]$. A sufficient condition is, setting $X=\lambda^{-1}$, to know all the $X>\frac{9}{4}$ satisfy
 \begin{multline}\label{ineq_qp}
 4\left(\frac{2}{3} \right)\left(\frac{4}{9}  \right)^p X^p+\frac{52\sqrt{2}\pi}{9}X \\
 \ge\zeta_{\Z^2}(4) X^4 +\frac{260\pi}{243}X^3+\left(\frac{104\sqrt{2}\pi}{27}+\frac{130\pi}{81} \right)X^2+E\left[\frac{1}{\sqrt{5}}\Z^2\right].
 \end{multline}
Defining the following coefficient
$$
\alpha_4:=\zeta_{\Z^2}(4),\quad \alpha_3:=\frac{260\pi}{243}, \quad \alpha_2:=\left(\frac{104\sqrt{2}\pi}{27}-\frac{130\pi}{81} \right),
$$ 
it follows by a direct estimate that \eqref{ineq_qp} holds if
 $$
X\geq U_p:=\max\left\{\left(\frac{9}{8}\left( \frac{9}{4} \right)^p \alpha_4  \right)^{\frac{1}{p-4}},  \left(\frac{9}{8}\left( \frac{9}{4} \right)^p \alpha_3  \right)^{\frac{1}{p-3}},\left(\frac{9}{8}\left( \frac{9}{4} \right)^p \alpha_2  \right)^{\frac{1}{p-2}}\right\}.
 $$
We observe that $U_p\to \frac{9}{4}$ as $p\to +\infty$ and is $U_p$ decreasing in $p$ for large $p$. Therefore, by continuity of $\lambda\mapsto E_f[\lambda \Z^2]$, for any $\varepsilon>0$, there exists $p_0$ such that
$$
\left| \min_{\lambda<4/9} E_f[\lambda \Z^2]- E_f\left[\frac{4}{9}\Z^2\right]\right|<\varepsilon.
$$
Since $ E_f[\frac{4}{9}\Z^2]>E_f[\frac{1}{\sqrt{5}}\Z^2]$, it follows that, for enough large $p$, 
$$
\min_{\lambda>0}E_f[\lambda \Z^2]=E_f[\frac{1}{\sqrt{5}}\Z^2]\approx -19.108745.
$$
The same argument can be repeated for $L=\mathsf{A}_2$, obtaining that for any $\varepsilon>0$, there exists $p_0$ such that for any $p>p_0$ there holds 
$$
\left| \min_{\lambda<4/9} E_f[\lambda \mathsf{A}_2]- E_f\left[\frac{4}{9}\mathsf{A}_2\right]\right|<\varepsilon.
$$
Therefore, by \eqref{mintri}, $\min_{\lambda>0} E_f[\lambda \mathsf{A}_2]> E_f[\frac{4}{9} \mathsf{A}_2]-\varepsilon>E_f[\frac{1}{\sqrt{5}}\Z^2]$ for $\varepsilon>0$ sufficiently small, which in turn is achievable for $p_0$ sufficiently large. These choices allow to complete the proof.
\end{proof}

\end{document}